\documentclass[10pt,journal,compsoc]{IEEEtran}
\usepackage{xcolor}
\usepackage{amsmath,amssymb,amsfonts}
\usepackage{amsthm}
\usepackage{thmtools}
\usepackage{thm-restate}
\usepackage{algorithm}
\usepackage{algorithmic}
\usepackage{graphicx}
\usepackage{subfigure}
\usepackage{float}
\usepackage{textcomp}
\usepackage{comment}
\usepackage{multirow}
\usepackage{threeparttable}
\newcommand{\tabincell}[2]{\begin{tabular}{@{}#1@{}}#2\end{tabular}}
\newtheorem{theorem}{Theorem}[section]
\newtheorem{lemma}{Lemma}[section]
\newtheorem{definition}{Definition}[section]
\newtheorem{prop}{Proposition}[section]

\ifCLASSOPTIONcompsoc
  \usepackage[nocompress]{cite}
\else
  \usepackage{cite}
\fi
\ifCLASSINFOpdf
\else
\fi
\begin{document}
%
\title{Fair and Differentially Private Distributed Frequency Estimation}
%
%
%
%

\author{Mengmeng Yang, \IEEEmembership{Member, IEEE,}
        Ivan Tjuawinata, \IEEEmembership{Member, IEEE,}
        Kwok-Yan Lam, \IEEEmembership{Senior Member, IEEE,}
        Tianqing Zhu, \IEEEmembership{Member, IEEE,}
        Jun Zhao, \IEEEmembership{Member, IEEE}
\IEEEcompsocitemizethanks{\IEEEcompsocthanksitem M. Yang, I. Tjuawinata, K-Y. Lam and J. Zhao are with the Strategic Centre for Research in Privacy-Preserving Technologies \& Systems, Nanyang Technological University, Singapore.\protect\\
E-mail: \{melody.yang, ivan.tjuawinata, kwokyan.lam, junzhao\}@ntu.edu.sg
\IEEEcompsocthanksitem T. Zhu is with University of Technology Sydney, Australia.
E-mail: Tianqing.Zhu@uts.edu.au}
}

%
%

\markboth{Journal of \LaTeX\ Class Files,~Vol.~14, No.~8, August~2015}%
{Shell \MakeLowercase{\textit{et al.}}: Bare Demo of IEEEtran.cls for Computer Society Journals}
%



\IEEEtitleabstractindextext{%
\begin{abstract}
In order to remain competitive, Internet companies collect and analyse user data for the purpose of improving user experiences. Frequency estimation is a widely used statistical tool which could potentially conflict with the relevant privacy regulations. Privacy preserving analytic methods based on differential privacy have been proposed, which either require a large user base or a trusted server; hence may give big companies an unfair advantage while handicapping smaller organizations in their growth opportunity. To address this issue, this paper proposes a fair privacy-preserving sampling-based frequency estimation method and provides a relation between its privacy guarantee, output accuracy, and number of participants. We designed decentralized privacy-preserving aggregation mechanisms using multi-party computation technique and established that, for a limited number of participants and a fixed privacy level, our mechanisms perform better than those that are based on traditional perturbation methods; hence, provide smaller companies a fair growth opportunity. We further propose an architectural model to support weighted aggregation in order to achieve higher accuracy estimate to cater for users with different privacy requirements. Compared to the unweighted aggregation, our method provides a more accurate estimate. Extensive experiments are conducted to show the effectiveness of the proposed methods.
\end{abstract}

}

\maketitle

\IEEEdisplaynontitleabstractindextext

%
\IEEEpeerreviewmaketitle

\IEEEraisesectionheading{\section{Introduction}\label{sec:introduction}}

%
%
%
%
\IEEEPARstart{T}{he} growth of the digital economy has led to a dramatic increase in the adoption of Internet applications and services used by billions of users. To gain the ability to sustain and grow in the face of fierce competition in the digital arena, applications and services providers need to be able to adapt their products to improve their users' experience. This can be done by analyzing their usage data using various data analytical and  statistical techniques. Frequency estimation is one of the basic statistical tools that can be used to analyse categorical data with a fixed domain space~\cite{cormode2019answering}. It provides companies with the distribution of the users' data which gives insights on the trend of the data as well as users' preferences. It has been widely used by various companies, Google analysts have used frequency estimation to investigate its users' web browsing behaviour~\cite{erlingsson2014rappor}. Furthermore, researchers from Apple have used frequency estimation to analyze users' preference on various system features such as emojis, health data types and media playback~\cite{DPTA}.

Frequency estimation has also been a fundamental building block of many other sophisticated data analysis and machine learning techniques such as ranking~\cite{DLSC15}, feature selection~\cite{Agh18}, and natural language processing~\cite{GDC12}. Despite the usefulness of frequency estimation and data analysis for decision-making process, if it is not done with care, such practice can violate the privacy of users' sensitive information, which has been a more important concern for users which is protected by various personal data and privacy protection regulations. To meet user's privacy expectation while providing accurate estimation, a well designed privacy-preserving mechanism should be used. 

Differential privacy (DP for short) is one of the \textit{de facto} privacy standards that has been used to measure the privacy level provided by various mechanisms to user data. 
Two types of differential privacy have been proposed, centralized DP and local DP. Local DP has attracted more attention recently and widely deployed by many big companies. Under local setting, each user perturbs his private data before sending it to the server. This means that no user's private data ever needs to leave the owner's device, which is a reasonable privacy feature that users desire when contributing their data. However, local DP mechanisms usually require a large number of participants for the analysis to produce useful result while maintaining a certain level of privacy guarantee to users~\cite{erlingsson2014rappor}. This gives big companies an unfair advantage in the competition in digital economy. Compared with local DP, centralized DP mechanisms do not need such a big number of participants. However, it assumes the existence of a trusted server which collects the users' data in the clear. When a query about the data set is received from a client, the trusted server can then perform the required operation to the data set together with the privacy-preserving measure to ensure that the output returned to the client is differentially private. However, the requirement of a server that is trusted by all users presents another hurdle for smaller organizations.

In this work, we aim to design a fair differentially private mechanism to solve the frequency estimation problem which can work well without requiring a large user base, hence can be useful for new and small companies with smaller user base. In essence, we consider a mechanism to be ``fair" when it provides a strong privacy guarantee as well as high estimation accuracy in the situation with a limited number of users. Such scheme enables smaller companies to perform user analysis and hence reduces the disadvantage gap for the growth opportunity of smaller companies.
More specifically, we consider a frequency estimation solution using sampling providing high accuracy output while maintaining the users' privacy without the need of a trusted server even in the situation with a limited number of users. We utilize multiparty computation techniques combined with DP mechanisms to achieve centralized differential privacy guarantee without the existence of any trusted server.
With regards to the DP mechanism that we focus on, we choose to focus on sampling method based on the following efficiency considerations. Sampling-based mechanisms are generally simpler than other DP mechanisms that are traditionally used in such application such as Gaussian mechanism. This leads to the former to generally have smaller computation and communication complexities. Furthermore, as we observed from the analysis done in this work, in some settings, when the privacy level is kept the same, this improvement in the efficiency is complemented with a better estimate accuracy. This motivates the study of sampling-based frequency estimation mechanisms. 

Considering the approach discussed above, in this work, we established a theoretical relation between privacy level, accuracy requirement, and number of users of  a classical sampling-based frequency estimation method. This relation shows that for a limited number of participants and a fixed privacy level, the sampling-based mechanism performs better than mechanisms that are based on more traditional perturbation methods such as Gaussian mechanism. This shows that to achieve the same level of privacy and accuracy settings, such solution requires less participants than its counterparts that are based on Gaussian mechanism. Furthermore, due to the less stringent requirement on the number of users compared to local DP mechanisms as well as the lesser computational demand of the sampling method, the use of MPC schemes becomes more feasible. With this consideration, we utilize MPC schemes to extend the sampling-based mechanism to a decentralized setting. In order to reduce the communication requirement imposed by the MPC scheme, we add another step of sampling such that users only report a fraction of the amount required by the initial mechanism. Combined with the observation we have made in the discussion of the first contribution, this shows that sampling-based differentially private mechanisms may be used to conduct data analysis when no trusted server is present, especially in scenarios where the number of users are limited. This leads to a stronger privacy guarantee compared to the original centralized setting since no confidential value needs to leave the data owners' device which provides a practical privacy-preserving aggregation technique to allow companies with smaller numbers of users to perform privacy-preserving data analytic. Lastly, we proposed an architectural model to support weighted aggregation in order to achieve a higher accuracy estimate in the scenario where users have different sensitivity towards the privacy of their data. The architectural model is designed to reduce the statistical error of the estimate by placing a larger priority to reports with statistically higher accuracy. Compared to a more conventional unweighted aggregation considered in previous investigations~\cite{akter2017computing}, our weighted aggregation method provides a statistically more accurate estimate. Extensive experiments have also been conducted to show the effectiveness of the proposed methods.

The rest of this paper is organized as follows. Section \ref{pre} introduces the preliminaries. The proposed solution is proposed in Section \ref{solu}. Section \ref{expe} and Section \ref{rela} show the experimental result and related work respectively while the paper is concluded in Section \ref{conc}. Due to the page limitations, some results are stated without proof. The proofs can be found in the Supplementary Material.

\section{Preliminaries}\label{pre}

\subsection{Problem Definition} 
In this paper, we consider the basic primitive that computes the item frequency.
Formally, let $U=\{u_1,\cdots, u_n\}$ be a set of users of size $n$, 
each user $u_i\in U$ has a value $v_i$ within a domain $\mathcal{I}=\{I_1,\cdots, I_N\}$ of size $N$ and reports it to an aggregator. 
The aggregator is interested in the number of users holding different items in $\mathcal{I}.$ 

Suppose that each user's private data can be efficiently and uniquely encoded to a data space $\mathcal{D}$ of size $N.$ 
Let $q$ be the smallest prime such that $q>n$ and denote by $\mathbb{F}_q$ the finite field containing $q$ elements. 
Then, we have $\mathcal{D}\subseteq \mathbb{F}_q^N.$ 
In our work, we encode user's value $v_i$ to a vector of length $N, \mathbf{e}_j\in \mathbb{F}_q^N$ which has value $1$ in its $j$-th entry and $0$ everywhere else. In addition, a default encoding without any private value is $\mathbf{0}\in \mathbb{F}_q^N,$ the zero vector of length $N.$ Hence we have $\mathcal{D}=\{\mathbf{0},\mathbf{e}_1,\cdots, \mathbf{e}_N\}.$
Besides, let $[N]$ denote the set $\{1, 2, \dots, N\}.$  We define $S_{N,M}=\{A\subseteq[N]:|A|=M\}$, which will be useful in our discussion of two-stage sampling. $M$ and $N$ are two positive integers. 

To prevent the user's data from being disclosed, we propose a randomized algorithm taking an encoded private data in $\mathcal{D}$ and outputs a report in a report space $\mathcal{R}$ for the reporting process to the aggregator. Such solution is used to collect the users' private value, through which, the data aggregator only learns an estimation for the number of users that hold different items in $\mathcal{I}$ without learning other information of the users' value while keeping the error of the estimations as small as possible. 

This work is based on the following assumptions
\begin{itemize}
    \item Secure communication channel between any pair of participants including server(s) and users. 
    \item Honest-but curious adversary; observes the data and communication of corrupted parties without maliciously changing any data.
    \item The set of corrupted parties contains at most one of the servers (no collusion amongst servers) or some users but not both.
\end{itemize}

\subsection{Differential Privacy}

Differential privacy is a privacy concept proposed by Dwork \textit{et al.} \cite{dwork2006differential} in 2006. 
It protects the users' private information by introducing some randomization defined as follows.

\begin{definition}[$(\epsilon, \delta)$-Differential Privacy \cite{dwork2014algorithmic}] 
Let $\mathcal{M}$ be a probabilistic algorithm, $\mathcal{M}:\mathcal{D}\rightarrow \mathcal{R}.$ $\mathcal{M}$ is $(\epsilon, \delta)$-differentially private if for all possible non-empty sets of outputs $S\subseteq \mathcal{M}(\mathcal{D})$ and for all neighbouring data set $D$ and $D'\in \mathcal{D},$ we have 
\begin{equation}
    Pr[\mathcal{M}(D)\in S] \leq e^\epsilon Pr[\mathcal{M}(D')\in S]+\delta.
\end{equation}
If $\delta=0$, we say that $\mathcal{M}$ is $\epsilon$-differentially private. 
\end{definition}

Intuitively, the definition states that adding, removing or changing a record in a data set 
cannot make a big difference to the final statistics.

\begin{definition}[$\epsilon$-Local Differential Privacy \cite{duchi2013local}] Let
$\mathcal{M}$ be a randomized algorithm, $\mathcal{M}:\mathcal{D}\rightarrow \mathcal{R}.$ $\mathcal{M}$ satisfies $\epsilon$-local differential privacy if and only if for any pair of distinct input values $v, v'\in \mathcal{D}$ and for any $\emptyset\subsetneq S\subseteq \mathcal{M}(\mathcal{D}),$ 
\begin{equation*}
Pr[\mathcal{M}(v)\in S] \leq e^\epsilon Pr[\mathcal{M}(v')\in S].
\end{equation*}
\end{definition}

Compared to centralized differential privacy, the perturbation mechanism $\mathcal{M}$ providing local differential privacy is applied to each data record independently. 

Sampling, generally used under the centralized setting with a trusted server storing the user's raw data, achieves $(\epsilon,\delta)$-differential privacy for mean value estimation~\cite{li2012sampling,sun2020federated}. The server randomly samples some of the records to estimate the overall statistics. To ensure privacy, the sampling result should not be disclosed along with the final statistics. In this work, we focus on sampling method in designing differentially private mechanisms over the more traditional techniques such as Laplace and Gaussian perturbations due to the following practical efficiency considerations. 
\begin{itemize}
    \item We use population and data sampling from either uniform or Bernoulli distribution, which is much simpler than sampling from more sophisticated distributions, such as Laplace or Gaussian. 
    \item Sampling can directly be applied to an $N$ dimensional object. In contrast, single dimensional Laplace and Gaussian perturbation need to be applied to each of the $N$ dimensions independently. 
    \item Population or data sampling produces a vector response over $\mathbb{F}_q$ while Laplace or Gaussian mechanism produces a real-number vector response. In this work, we extend the sampling-based mechanism to a decentralized setting via secret sharing. This can be done to a sampling-based mechanism without changing the underlying field while maintaining information-theoretical security. On the other hand, an information-theoretically secure decentralized protocol based on Gaussian mechanism or Laplace mechanism requires further encoding to a finite field element using techniques such as fixed-point encoding~\cite{OS10}. This leads to a much larger finite field, which increases the storage, commuication, and computation requirements for each user.  
\end{itemize}

We also proved that our mechanism satifies $(\epsilon,\delta)-$ differential privacy and produces a more accurate estimate over smaller population size compared to the estimate produced by the Gaussian mechanism, which is a more traditional mechanism that achieves $(\epsilon,\delta)-$ differential privacy. 

If any part of the mechanism only depends on the output of other mechanisms without any direct dependence on the private data, such sub-protocol can be seen as post-processing and it does not consume any privacy budget, as formalized in Proposition~\ref{prop:postprocessing}.

\begin{prop}[Post-Processing \cite{dwork2014algorithmic}]\label{prop:postprocessing}
Let $\mathcal{M}:\mathcal{D} \rightarrow \mathcal{R}$, be a randomized algorithm that is $(\epsilon, \delta)$-differentially private.
Let $f$ be an arbitrary randomized mapping, $f:\mathcal{R} \rightarrow \mathcal{R}' $. Then $f\circ M: \mathcal{D} \rightarrow \mathcal{R}'$ is $(\epsilon, \delta)$-differentially private.
\end{prop}

\subsection{Additive Secret Sharing}
Secret sharing scheme (SSS for short) is a privacy-preserving technique that is designed to enable a dealer with a secret value $s$ to distribute it to a group of participants in the form of shares. It was first independently proposed by Shamir and Blakley in 1979. The share is probabilistically generated from the secret such that sufficient number of the shares can be used to recover the original shares. A family of secret sharing schemes that is widely used is additive secret sharing scheme. The formal definition of additive secret sharing scheme is presented as follows. 
\begin{definition}[Additive Secret Sharing Scheme]
Let $m$ be a positive integer and $\mathbb{F}$ be a finite field. Given a secret value $s\in \mathbb{F},$ an additive secret sharing scheme with $m$ parties over $\mathbb{F}$ generates the $m$ shares in the following way. First, we independently sample $s_1,\cdots, s_{m-1}$ from $\mathbb{F}$ uniformly at random. Having $s_1,\cdots, s_{m-1}$ and $s,$ define $s_m=s-(s_1+\cdots+s_{m-1}).$  We define the $m$ shares to be $(s_1,\cdots, s_m).$ 
\end{definition}
It is easy to see that with the knowledge of all $m$ shares, we can recover $s$ by calculating $s_1+\cdots+s_m.$ However, if an adversary learns at most $m-1$ of the shares, the distribution of $s$ is still uniform, independent of the knowledge of the $m-1$ revealed shares. In other words, he does not learn any information about $s.$ 

A property of additive secret sharing scheme that can be easily verified and essential in this paper is its additive property summarised in the following proposition.

\begin{prop}
Suppose that we are using an additive secret sharing scheme with $m$ parties over a finite field $\mathbb{F}.$ For two secret values $a,b\in \mathbb{F},$ suppose that $a$ and $b$ are additively secretly shared with shares $(a_1,\cdots, a_m)$ and $(b_1,\cdots, b_m)$ respectively. Then $(a_1+b_1,a_2+b_2,\cdots, a_m+b_m)$ gives a valid additive secret sharing of $a+b.$
\end{prop}

The additive property shows that to find a valid additive secret sharing of the sum of $M$ private values, we can additively secret share each private value to the $m$ parties. Each party can then sum up the $M$ shares he receives to get a share for the sum of the $M$ private values. It can then be shown, by the use of simulation-based security, that such additive secret-sharing based aggregation method provides information theoretical security against an honest-but-curious adversary controlling up to $m-1$ of the computing parties. A more detailed discussion of the simulation-based security and adversarial setting can be found, for example, in~\cite{Lindell20}. Such security guarantee shows that by the use of the additive secret-sharing based aggregation method discussed above, we can simulate the existence of a trusted server who assists the users in the calculation with perfect security without actually needing any such trusted server. This shows that we can aim for a centralized differential privacy guarantee without the need of revealing any private data to other entities.

\section{Proposed solution}  \label{solu}

In this section, we present the proposed solution of privacy-preserving frequency estimation to mitigate the fairness gap between companies with different sizes. Our solution enables organizations with small user size to still perform accurate statistical analysis to their users' data without the need of any users' data to leave their respective devices.
The flow of the proposed methods are shown as follows.

\begin{enumerate}
\item Users are partitioned into several groups according to their privacy preferences. Because of the simplicity of this step, we assume this is done before the protocol and omit the related discussion. 
\item Each privacy group conducts a privacy-preserving frequency-estimation mechanism following the preferred privacy level independently to obtain an estimate for the distribution of the data held by users in the group.
\item Having the estimates of the data distribution from all privacy groups, the server performs weighted aggregation to obtain the estimate of the overall data distribution. This is done using the different weights assigned to the respective privacy groups.
\end{enumerate}

In the following sections, we discuss the mechanisms in Steps 2 and 3 in more detail.

\subsection{Differential Privacy with Distributed Sampling} 
We consider a privacy-preserving solution to the problem of frequency-estimation that is based on sampling method. First, we analyze a common way sampling can be used to design a privacy-preserving frequency estimation mechanism, which is then extended to a decentralized setting. To reduce the communication requirement incurred by the users due to the use of secret sharing, we propose a variant which we name two-stage sampling method.

\subsubsection{General Sampling}

By general sampling, we refer to the method that calculates the frequency estimate from the private values of a randomly chosen subset of the users. 
Specifically, under the centralized setting, the server randomly selects part of the records and perform the statistics over the selected records. 
Algorithm~\ref{Funct:RSwASS} provides the specification of the general sampling-based frequency estimation mechanism, which is denoted by $\mathtt{DPCS}.$

\begin{algorithm}[!htbp]
\caption{Differential privacy with centralized sampling ($\mathtt{DPCS}$)}\label{Funct:RSwASS}
\hspace*{0.02in} {\bf Input:}
The trusted server holds the encoded private values of $n$ users $\mathbf{t}_i\in\{\mathbf{e}_1,\cdots,$  $\mathbf{e}_N\}\subseteq \mathcal{D},$ sampling probability $p,$ privacy budget $\epsilon;$\\
\hspace*{0.02in} {\bf Output:}
Estimate of normalized frequency of each value $\widehat{\boldsymbol \psi};$
\begin{algorithmic}[1]
\STATE The trusted server samples from the $n$ users to choose the participants where each user has probability $p$ to be chosen, let $\mathcal{A}\subseteq U$ be the set of participating users;
\STATE The trusted server computes $\widehat{\boldsymbol \tau} =\sum_{u_i\in \mathcal{A}} \mathbf{t}_i;$ 
\STATE Server computes and outputs the normalized frequency estimation $\widehat{\boldsymbol \psi}=\frac{1}{pn} \widehat{\boldsymbol\tau}$ as a vector over real numbers;
\end{algorithmic}
\end{algorithm}

\noindent\textbf{Privacy analysis:}

We show that $\mathtt{DPCS}$
provides an $(\epsilon,\delta)$ differential privacy, which is discussed in Theorem~\ref{lem1}.
\begin{restatable}{theorem}{lemone}\label{lem1}
Let $\epsilon,\delta>0$ be given. Suppose that there exists a positive real number $\beta$ such that for any $i=1,\cdots, N,$ there are at least $\beta n$ users owning item $I_i$ and
\begin{equation}\label{bound}
\beta\geq \frac{1}{2\pi n (e^{-\epsilon}-e^{-2\epsilon})}\max\left(\left(\frac{2\pi}{\delta}\right)^{\frac{2}{N+1}},\left(\frac{1}{\delta}\right)^{\frac{2}{N}}\right).
\end{equation}
 Then, setting $p= 1-e^{-\epsilon},$ the mechanism $\mathtt{DPCS}$ provided in Algorithm \ref{Funct:RSwASS} is $(\epsilon,\delta)$-differentially private.
\end{restatable}

\textit{Discussion.} According to Theorem \ref{lem1}, we find that the privacy level is not only determined by the privacy parameters, but also population size and data dimension. Theorem \ref{lem1} provides a more relaxed bound for the privacy that no matter how the data set changing, as long as it satisfies the Eq. \eqref{bound}, it provides a fixed privacy guarantee for a fixed sampling probability $p.$ To better understand how the sampling mechanism affect the privacy guarantee, we further provide a tighter bound as shown in Theorem \ref{lem2}. For the same sampling probability, Theorem \ref{lem2} provides a stronger privacy guarantee. However, we observe that the privacy depends on a stricter data distribution requirement. 

\begin{restatable}{theorem}{lemtwo}\label{lem2}
Let $\epsilon,\delta>0$ be given. Suppose that there exist two positive real numbers $\beta$ and $z$  such that for any $i=1,\cdots, N,$ there are at least $\beta n$ users owning item $I_i$ where we require $\beta$ and $z$ to satisfy the following requirements. Firstly, we require that
\[\beta\geq \frac{e^{z+\epsilon}}{2\pi n (1-e^{-\frac{\epsilon}{2}})^2}\max\left(\frac{\left(\frac{2\pi}{\delta}\right)^{\frac{2}{N+1}}}{(1+e^{-\frac{\epsilon}{2}})^2},\left(\frac{4\pi}{\delta}\right)^{\frac{2}{N}}\right).\]
Furthermore, we also require that
\begin{equation*}
  \nonumber z\leq \ln(1+n(1-\beta(N-1))(1-e^{-\frac{\epsilon}{2}})).  
\end{equation*}
Then, setting $p= 1-e^{-z-\epsilon},$ the mechanism $\mathtt{DPCS}$ in Algorithm \ref{Funct:RSwASS} preserves $(\epsilon,\delta)$-differential privacy.
\end{restatable}

\noindent\textbf{Utility analysis:}

Here we provide some statistical analysis of $\widehat{\boldsymbol\psi}$ to show the accuracy of the sampling.

\begin{restatable}{lemma}{lemufDPDS}\label{ufDPDS}

Suppose that for $i=1,\cdots, N,$ there are $\Pi_i$ users whose true value is $\mathbf{e}_i.$ In other words, the value we want to achieve is ${\boldsymbol \psi}=\left(\frac{\Pi_1}{n},\cdots, \frac{\Pi_N}{n}\right).$ Then $\widehat{\boldsymbol \psi}$ is an unbiased estimator of ${\boldsymbol\psi}$ and $\mathrm{Var}\left(\widehat{\boldsymbol\psi}\right)= \frac{1-p}{pn^2} V$ where $V$ is a diagonal matrix of size $N\times N$ with the entry in row $i$ column $i$ being $\Pi_i.$ Furthermore, we have the expected square $L_2$ distance of $\widehat{\boldsymbol \psi}$ to the true value ${\boldsymbol\psi}$ to be $\frac{1-p}{pn},$ that is,

\begin{equation*}
    \mathbb{E}\left(\left\| \widehat{\boldsymbol\psi}-{\boldsymbol\psi}\right\|_2^2\right)=\frac{1-p}{pn}.
\end{equation*}

\end{restatable}

\subsubsection{General Sampling (Decentralized)}

We proceed by designing the protocol $\mathtt{DPDS}$ which realizes $\mathtt{DPCS}$ in a decentralized setting. The full specification of the protocol $\mathtt{DPDS}$ can be found in Algorithm~\ref{RSwASS}.

\begin{algorithm}[!htbp]
\caption{Differential privacy with distributed sampling ($\mathtt{DPDS}$)}\label{RSwASS}
\hspace*{0.02in} {\bf Input:}
Encoded private values of $n$ users $\mathbf{t}_i\in\{\mathbf{e}_1,\cdots,$  $\mathbf{e}_N\}\subseteq \mathcal{D},$ sampling probability $p,$ privacy budget $\epsilon;$\\
\hspace*{0.02in} {\bf Output:}
Estimate of normalized frequency of each value $\widehat{\boldsymbol \psi};$
\begin{algorithmic}[1]
\STATE \textit{[User side]}
\STATE User $u_i$ samples a random Bernoulli random variable with success probability $p$ to represent whether he actively participates on the count;
\IF{$u_i$ actively participates,}
\STATE $u_i$ encodes her response as $\widehat{\mathbf{t}}_i = \mathbf{t}_i\in \mathcal{D}$;
\ELSE 
\STATE $u_i$ encodes her response as $\widehat{\mathbf{t}}_i = \mathbf{0}\in \mathcal{D}$;
\ENDIF
\STATE $u_i$ samples $n-1$ random vectors of length $N,\widehat{\mathbf{s}}_{i,1},$ $\cdots, \widehat{\mathbf{s}}_{i,n-1}\in \mathbb{F}_q^N$ and sets $\widehat{\mathbf{s}}_{i,n}=\widehat{\mathbf{t}}_i-\sum_{j=1}^{n-1} \widehat{\mathbf{s}}_{i,j};$
\STATE $u_i$ sends $\widehat{\mathbf{s}}_{i,j}$ to $u_j$ for $j=1,\cdots, n;$
\STATE $u_i$ computes $\widehat{\boldsymbol \tau}_i=\sum_{j=1}^n \widehat{\mathbf{s}}_{j,i}$ and sends it to the Server;
\STATE \textit{[Server Side]}
\STATE Server computes $\widehat{\boldsymbol \tau}=\sum_{i=1}^n \widehat{\boldsymbol \tau}_i\in \mathbb{F}_q^N$ and regard the sum as a vector over real numbers of length $N;$
\STATE Server computes and outputs the normalized frequency estimation $\widehat{\boldsymbol \psi}=\frac{1}{pn} \widehat{\boldsymbol\tau}$ as a vector over real numbers;
\end{algorithmic}
\end{algorithm}

As shown in Algorithm \ref{RSwASS}, user's value is represented using one hot encoding. Each user decides his extent of participation in the reporting process with probability $p$ (Line 2). Here $p$ indicates the sampling probability, which determines the privacy level. 
If the user decides to actively participate, the user keeps his true value. If the user does not actively participate, his value is perturbed as a zero vector with length $N$ (Lines 3-7). 
The secret-sharing based aggregation scheme discussed before is then used with the $n$ users as the $n$ computing parties before sending the shares of the aggregated value to the server which serves as the aggregator (Lines 8-10).
The server estimates the frequency for each item after collecting reports as $\widehat{\boldsymbol \psi}=\frac{1}{pn}\sum_{i=1}^n \widehat{\boldsymbol \tau}_i$ (Line 12). 

\noindent\textbf{Privacy and utility analysis:}

To better facilitate the extension from the centralized setting to a decentralized one, we introduce a small modification to $\mathtt{DPCS},$ which we will denote by $\mathtt{DPCS}^*.$ The only difference between $\mathtt{DPCS}$ and $\mathtt{DPCS}^*$ lies on what the mechanisms output. Instead of outputting $\widehat{\boldsymbol \psi}$ directly as $\mathtt{DPCS}$ does, the trusted server secretly shares $\widehat{\boldsymbol \tau}$ to $n$ vectors of length $N$ over $\mathbb{F}_q.$ Having these shares, a receiver will recover the value of $\widehat{\boldsymbol \tau},$ regards it as a vector over real numbers and performs the normalization by calculating $\widehat{\boldsymbol \psi} = \frac{1}{pn} \widehat{\boldsymbol \tau}.$ Note that with this modification, the only possible extra information that the receiver may get is the number of users, which is $n.$ However, such information is publicly known. Hence $\mathtt{DPCS}^*$ has the same privacy and accuracy guarantee as $\mathtt{DPCS}$ and the same proofs from Theorems~\ref{lem1} and ~\ref{lem2} as well as Lemma~\ref{ufDPDS} are also applicable to $\mathtt{DPCS}^*.$ It is easy to see that from the perspective of the receiver, $\mathtt{DPDS}$ is indistinguishable from $\mathtt{DPCS}^*.$ Hence, the privacy and accuracy guarantees of $\mathtt{DPCS}^*$ also apply for $\mathtt{DPDS}.$
In addition, since we are using additive secret sharing scheme with $n$ parties, we can guarantee that any collusion of up to $n-1$ parties learns no information regarding the private value of the remaining party. Hence $\mathtt{DPDS}$ provides privacy against an honest-but-curious adversary controlling either up to $n-1$ users or the server but not both.

\noindent \textbf{Complexity analysis:}

It is easy to see that throughout $\mathtt{DPDS},$ each user sends $2nN+N$ elements of $\mathbb{F}_q$ for his communication cost. In terms of computation cost, each user needs to perform $1$ Bernoulli sampling, $(n-1)N$ uniform sampling from $\mathbb{F}_q$ and $2(n-1)N$ field addition operations. On the other hand, the server needs no communication cost and $(n-1)N$ field addition operations and $N$ real number multiplication operations.

\noindent \textbf{Comparison against Distributed Gaussian Mechanism:}

Recall that our mechanism achieves $(\epsilon,\delta)-$differential privacy. A comparable mechanism that can also provide the same privacy guarantee is the Gaussian mechanism. In this section, we provide a discussion on the extension of Gaussian-based frequency estimation to a distributed setting using secret sharing. The perturbation can then be done by letting each user to perform a piece of Gaussian sampling before performing the secret sharing based aggregation.
Recall that for a data set containing data of $n$ users with $\Pi_i$ users holding item $I_i$ for $i=1,\cdots, N,$ the output we want to obtain is $\Psi=\frac{1}{n}\left(\Pi_1,\cdots, \Pi_N\right).$ It is then easy to see that the sensitivity of the query is $G_f=\frac{\sqrt{2}}{n}.$ Hence, to obtain $(\epsilon,\delta)$-differential privacy, the data can be perturbed by a Gaussian noise sampled from $\mathcal{N}(0, G_f^2\sigma_G^2)$ where 
$\sigma_G = \frac{\sqrt{2\ln\left(\frac{1.25}{\delta}\right)}}{\epsilon}$. By the additivity of identically and independently distributed random variables following a fixed Gaussian distribution, such noise can be generated by letting each user add a Gaussian noise sampled from $\mathcal{N}(0,(\frac{G_f\sigma_G}{n})^2)$, performing the additive secret sharing to distribute the shares among the $n$ users and revealing the aggregated result to the server. We denote such mechanism by $\mathtt{DPDG}.$ It is easy to see that $\mathtt{DPDG}$ is $(\epsilon,\delta)$-differentially private which outputs an unbiased estimator of $\Psi$ with variance $V_G=\frac{4\ln\left(\frac{1.25}{\delta}\right)}{n^2\epsilon^2}.$ In the following, we discuss the advantages and disadvantages of $\mathtt{DPDG}$ compared to $\mathtt{DPDS}.$

Compared to $\mathtt{DPDS}$ which requires $\Pi_i\geq \beta n$ for some $\beta\in(0,1)$ to ensure its privacy, $\mathtt{DPDG}$ does not have any of such requirement on the actual data distribution. Furthermore, it provides smaller variance, $V_G<V,$ when $n$ is sufficiently large, i.e., when $n>\frac{4(e^\epsilon-1)\ln\left(\frac{1.25}{\delta}\right)}{\epsilon^2}$. This shows that in situations with smaller population size, $\mathtt{DPDS}$ generally performs better compared to $\mathtt{DPDG}.$ Furthermore, $\mathtt{DPDG}$ generally requires higher complexity in terms of both computation and communication.
Compared to $\mathtt{DPDS},$ the variant $\mathtt{DPDG}$ replaces $1$ Bernoulli sampling outputting a bit with $1$ multi-dimensional Gaussian sampling outputting a vector of length $N$ with each entry being a real number. Since this Gaussian sampling outputs a vector of real numbers, we cannot directly store any element as a field element. Secret sharing schemes over real numbers can generally be defined by encoding real numbers as field elements via fixed point arithmetic~\cite{OS10}. For $k$-bit decimal points accuracy, the field size needs to increase by at least $k$ bits, making both computational and communication complexity to achieve such accuracy to be much larger than those of $\mathtt{DPDS}.$ 

\subsubsection{Two-Stage Sampling}
Although the mechanism $\mathtt{DPDS}$ discussed in the previous section provides high accuracy estimation of the frequency calculation, it requires the distribution of $n$ vectors of length $N$ to be $n$ different users. Such requirement may become infeasible when the number of users grows. To reduce the communication cost, a natural solution is to reduce the amount of items to report as well as the number of shares being generated for each of such shares. Following this observation, we propose a two-stage sampling method. Intuitively, after performing the same sampling method as has been done in $\mathtt{DPDS}$ to decide the extent of participation each user will do, each user performs a second round of sampling to decide the items he is going to report. Specifically, we consider two types of sampling for the second sampling process. 

\noindent\textit{Uniform sampling}. Here, to determine which items to sample, we fix the number of items to report and each item is sampled uniformly at random. By only reporting some of the values instead of the whole $\widehat{\mathbf{t}}_i\in \mathcal{D},$ the overall communication and computational cost can be significantly reduced. However, this comes with a decrease of statistical accuracy. Formally, uniform sampling process is used to uniformly sample $\alpha N$ items out of the $N$ items for each user to report to for a predetermined $\alpha.$ In other words, for any user, for any $A\in S_{N,\alpha N},$ the probability that a user reports the $\alpha N$ items corresponding to elements of $A$ is $\frac{1}{\binom{N}{\alpha N}}.$ We define such reporting distribution to be $\chi_U.$

\noindent\textit{Adaptive sampling}. In order to limit the decrease of the statistical accuracy from reporting less values, the sampling probability may be changed. Instead of sampling the values uniformly, we increase the probability that for actively participating user to report the value corresponding to the actual item he is holding. Specifically, we adopt the idea we proposed in~\cite{9264723} for the item sampling. For a fixed constant $\alpha\in(0,1),$ the sampling process generates a random subset of $[N]$ of size $\alpha N$ with sets containing the user's true value having a larger probability to get sampled. 
We formally define the adaptive sampling process shown as follows.
    \begin{definition}[Adaptive Sampling]
    Let $N$ be a positive integer, $\alpha\in(0,1)$ and $\gamma>1.$ For $i=1,\cdots, n,$ let $Y_i^{(N,\alpha,\gamma)}$ be a random variable with values from $S_{N,\alpha N}$ that represents the set of items that user $u_i$ reports. We suppose further that after the first stage of sampling process, we have encoded the response of $u_i$ to be $\widehat{\mathbf{t}_i}\in \mathcal{D}.$ Then for $A\in S_{N,\alpha N},$ we define $Pr(Y_i^{(N,\alpha,\gamma)}=A)\triangleq\zeta_A$ where
    \[
    \zeta_A=\left\{
    \begin{array}{cc}
    \frac{1}{\binom{N}{\alpha N}},     &\mathrm{~if~}\widehat{\mathbf{t}_i}=\mathbf{0},  \\
    \frac{\gamma}{\gamma\binom{N-1}{\alpha N-1}+\binom{N-1}{\alpha N}},     &\mathrm{~if~}\mathbf{v}=\mathbf{e}_j\mathrm{~and~}j\in A,\\
    \frac{1}{\gamma\binom{N-1}{\alpha N-1}+\binom{N-1}{\alpha N}},     &\mathrm{~if~}\mathbf{v}=\mathbf{e}_j\mathrm{~and~}j\notin A.\\
    \end{array}
    \right.
    \]
    We denote such reporting distribution by $\chi_A.$
    \end{definition}

\textbf{Remark.} The second stage of sampling may reduce the amount of communication and computation needed by each user. However, it also comes with some drawbacks. Firstly, there will need to be a preliminary report from each user to identify the items he will be reporting. This is to ensure that shares are appropriately labeled and aggregation process is not started before all the shares are sent. Furthermore, such process will also leak information. More specifically, due to the preliminary report, the server can identify a smaller group of users where the reported count comes from. In some cases, such additional information may be used to infer more information about users' private values. To avoid such additional information from being leaked, we introduce another server whose role is to receive such preliminary report and perform further processing to reduce the amount of information leakage from the aggregating server. We denote such scheme by $\mathtt{TSS}$ and its full specification can be found in Algorithm~\ref{2SSP}.

\begin{algorithm}[!htbp]
\caption{Two-stage sampling ($\mathtt{TSS}$)}\label{2SSP}
\hspace*{0.02in} {\bf Input:}
Private value of $n$ users $\mathbf{t}_i\in\{\mathbf{e}_1,\cdots, \mathbf{e}_N\}\subseteq \mathbb{F}_q^N,$ sampling probability $p,$ reporting proportion $\alpha,$ reporting distribution $\chi\in\{\chi_U,\chi_A\}$ (with possible reporting parameter $\gamma$ if necessary) and privacy budget $\epsilon;$\\
\hspace*{0.02in} {\bf Output:}
Estimate of normalized frequency of each value $\widehat{\boldsymbol \psi};$
\begin{algorithmic}[1]
\STATE \textit{[User sides]}
\STATE User $u_i$ samples a random Bernoulli random variable with success probability $p$ to represent whether he actively participates on the count;
\IF{$u_i$ actively participates,}
\STATE $u_i$ encodes his response as $\widehat{\mathbf{t}}_i = \mathbf{t}_i\in \mathcal{D}$;
\ELSE 
\STATE $u_i$ encodes his response as $\widehat{\mathbf{t}}_i = \mathbf{0}\in \mathcal{D}$;
\ENDIF
\STATE $u_i$ samples $A_i\in S_{N,\alpha N}$ using distribution $\chi$ with possible implicit parameters $\gamma$ and $\widehat{\mathbf{t}_i};$
\STATE $u_i$ defines a binary vector $\mathbf{a}_i$ where its $j$-th entry is $1$ if $j\in A_i$ and $0$ otherwise;
\STATE \textit{[$\mathcal{S}_1$ Side]}
\STATE $\mathcal{S}_1$ compiles $\mathbf{a}_1,\cdots, \mathbf{a}_n$ to obtain the numbers $m_1,$ $\cdots, m_N$ where $m_j$ is the number of users reporting for the count of item $I_j;$
\FOR{$j=1,\cdots, N$}
\STATE $\mathcal{S}_1$ randomly selects $m_j$ users $u_{j,1},\cdots, u_{j,m_j};$
\STATE $\mathcal{S}_1$ publishes $(j,u_{j,1},\cdots, u_{j,m_j})$ to all users and the aggregating server $\mathcal{S}_2;$
\FOR{$i$ such that $j\in A_i$}
\STATE $u_i$ with encoded response for item $I_j, \widehat{t_{i,j}}$ samples $m_j-1$ random elements of $\mathbb{F}_q,$ denoted by $s_{j,i,1},\cdots,s_{j,i,m_j-1}\in \mathbb{F}_q$ and sets $s_{j,i,m_j}= \widehat{t_{i,j}}-\sum_{t=1}^{m_j-1} s_{j,i,t}\in \mathbb{F}_q;$
\STATE $u_i$ sends $(j,s_{j,i,t})$ to $u_{j,t}$ for $t=1,\cdots, m_j;$
\ENDFOR
\STATE \textit{[For elected users $u_{j,1},\cdots, u_{j,m_j}$]}
\STATE Upon receiving $m_j$ values $s_{j,i,t}$ for $i$ such that $j\in A_i,$ $u_{j,t}$ computes $s_{j,t}=\sum_{i:j\in A_i} s_{j,i,t}$ and sends $(j,s_{j,t})$ to $\mathcal{S}_2;$
\STATE \textit{[$\mathcal{S}_2$ Side]}
\STATE $\mathcal{S}_2$ calculates the estimate of the count for item $j$ as $s_j=\sum_{t=1}^{m_j} s_{j,t}$ and treat it as a real number;
\STATE $\mathcal{S}_2$ performs some post-processing to obtain the unbiased estimator for the count of item $I_j$ as $\widehat{\psi_j}=\frac{1}{q_\chi n}s_j$ where $q_\chi=p p_\chi$ and
\[p_\chi=\left\{
\begin{array}{cc}
    \alpha, &\mathrm{~if~}\chi=\chi_U  \\
     \frac{\alpha\gamma}{\alpha\gamma+1-\alpha},&\mathrm{~if~}\chi=\chi_A 
\end{array}
\right. .
\]

\ENDFOR
\end{algorithmic}
\end{algorithm}

As shown in Algorithm \ref{2SSP}, instead of reporting the entire encoded value, users report the set of items they intend to report to server $\mathcal{S}_1$ after the encoding and sampling processes (Lines $2-9$). So the only information about the users' private values obtained by $\mathcal{S}_1$ is the sets $A_i$ containing the $\alpha N$ items user $u_i$ intends to report which provides the number of users that report different items $I_j$ (Line $11$). 
For $j=1,\cdots, N,$ once the number $m_j$ of users that report item $I_j$ is obtained, server $\mathcal{S}_1$ randomly selects $m_j$ users to help in the reporting process of item $I_j$ to $\mathcal{S}_2$ and inform all participants about the identity of the elected users  (Lines $13$ and $14$). 
The original users who intend to report to item $I_j$ secretly share their encoded response to item $I_j$ to the $m_j$ elected users using additive secret sharing (Lines $15-18$). 
The elected users then aggregate the shares they receive corresponding to item $I_j$ and send the sum to the server $\mathcal{S}_2$ (Lines $20$). 
After receiving all $m_j$ reports from the $m_j$ elected users, server $\mathcal{S}_2$ estimates the distribution of item $I_j$ for all $j=1,\cdots, N$ (Lines $22-23$). 
Through the help of $\mathcal{S}_1$ and the elected users, although server $\mathcal{S}_2$ may learn $m_j,$ the total number of users that may affect the count $s_j,$ the identities of these $m_j$ users are hidden. In the following, we analyze the privacy protection guarantee of $\mathtt{TSS}$ against the two servers. 

\noindent\textbf{Privacy Analysis:}

Recall that server $\mathcal{S}_1$ receives the user's indicator for which items to report. So the information $\mathcal{S}_1$ receives only contains the items each user is intending to report without having further information on the actual report. Intuitively, such information provides no information regarding the users' private value if the reporting distribution $\chi$ is uniform. However, if the indicator is selected adaptively, this information may leak some information about his private value. This is due to the distribution design having sets that contain his actual private value to be more likely to be sampled. As observed in Theorem~\ref{s1lemme}, however, such privacy leak can be well bounded.

\begin{restatable}{theorem}{sonelemme}\label{s1lemme}
Let $\mathcal{M}$ be the mechanism following $\mathtt{TSS}$ from the point of view of $\mathcal{S}_1.$ In other words, for each user, the mechanism takes $\mathbf{v}_i$ as input and outputs $A_i.$ Let $\alpha$ be the reporting proportion for the second stage while $\chi_U$ and $\chi_A$ represent uniform and adaptive distribution for the reporting distribution respectively. Then the mechanism $\mathcal{M}$ provides an $\epsilon^*$ local differential privacy against $\mathcal{S}_1$ where 
\[
\epsilon^*=\left\{
\begin{array}{cc}
0, &\mathrm{~if~}\chi=\chi_U;\\
\log \gamma,&\mathrm{~if~}\chi=\chi_A;
\end{array}
\right.
\] 
\end{restatable}

Theorem \ref{s1lemme} states that for uniform sampling, the user's information is indeed perfectly hidden to $\mathcal{S}_1.$ In the situation where adaptive sampling is used, on the other hand, as the user's true item has a higher chance to be selected, such information can be used to infer some statistical information on the user's true item. However, the sampling process is shown to satisfy local differential privacy definition where the privacy disclosure is bounded by $\log \gamma.$ In the following, we consider the privacy provided by $\mathtt{TSS}$ against $\mathcal{S}_2.$  
\begin{restatable}{theorem}{lemTSStwo}\label{lem:2SSP>DPDS}
The mechanism $\mathtt{TSS}$ provides at least the same privacy level as $\mathtt{DPDS}$ against $\mathcal{S}_2$ given the same sampling probability regardless of the selected reporting distribution $\chi.$
\end{restatable}

\noindent \textbf{Utility Analysis:}

Consider the random variable $\widehat{\psi_j}$ that estimates the actual count of item $I_j.$ We note that since it is impossible for holders of any other items to affect the count $s_j,$ we can consider each random variable $\widehat{\psi_j}$ independently of each other. 
Lemma \ref{lem:probuseful} shows the distribution of the final estimation.

\begin{lemma}\label{lem:probuseful}
Fix $i\in\{1,\cdots, N\}$ and suppose that a user $u_j$ holds item $I_i$ and he decides to actively participate in the $\mathtt{TSS}$ mechanism. Suppose further that he uses the distribution $\chi$ to determine $A_j,$ the $\alpha N$ items he will report to. Then the probability that $i\in A_j$ is $p_\chi$ where
$$p_\chi=\left\{
\begin{array}{cc}
\alpha,&\mathrm{~if~} \chi=\chi_U\\
\frac{\alpha\gamma}{\alpha\gamma+1-\alpha},&\mathrm{~if~} \chi=\chi_A
\end{array}
\right. .$$
Hence, since the probability that a user $u_j$ actively participates in $\mathtt{TSS}$ is $p,$ we have that for any user $u_j$ that holds item $I_i,$ the probability that he can affect the count of $s_j$ is $q_\chi\triangleq p_\chi\cdot p.$

\end{lemma} 

Such distribution can then be used to provide the statistical accuracy analysis for $\mathtt{TSS}.$

\begin{theorem}\label{lem:2SSPUnbiased}
For $\mathtt{TSS},$ suppose that for $i=1,\cdots, N,$ there are $\Pi_i$ users having the value $\mathbf{e}_i.$ In other words, the value we want to achieve is ${\boldsymbol \psi}=\left(\frac{\Pi_1}{n},\cdots, \frac{\Pi_N}{n}\right).$ Then $\widehat{\boldsymbol \psi}$ is an unbiased estimator of ${\boldsymbol\psi}$ and $\mathrm{Var}\left(\widehat{\boldsymbol\psi}\right)=\frac{1-q_\chi}{nq_\chi}.$
\end{theorem}
\begin{proof}
This lemma can be shown using exactly the same proof structure as Lemma~\ref{ufDPDS} with the only difference being that for any $j=1,\cdots, N,$ the binomial random variable $\widehat{\psi_j}$ has a success probability of $q_\chi$ instead of $p.$ So following the proof idea from before, we obtain the desired claim. 
\end{proof}

\noindent \textbf{Complexity Analysis:}

Lastly, we discuss the computational and communication costs for different groups of participants of $\mathtt{TSS}.$ A summary of such requirements as well as its comparison with the requirement for $\mathtt{DPDS}$ can be found in Table~\ref{comparison}. A detailed discussion on the calculation of the complexity of $\mathtt{TSS}$ can be found in the Appendix. As can be observed, although we see an increase in complexity for the server, especially for the newly introduced server $\mathcal{S}_1,$ the average computational and communication complexity costs incurred by an average user is reduced approximately by a factor of $\frac{1}{p_\chi \alpha},$ achieving the objective of introducing the second stage of sampling.
\begin{table*}[ht]
    \centering
    \begin{threeparttable} 
    \begin{tabular}{|c|c|c|c|c|c|}
    \hline
         & \multicolumn{2}{|c|}{$\mathtt{DPDS}$} & \multicolumn{3}{|c|}{$\mathtt{TSS}$ } \\
         \hline
        & Server & User & $\mathcal{S}_1$ & $\mathcal{S}_2$ & User \\
        \hline
        Computation &\tabincell{l}{FA:($(n-1)N$)\\ RM: $N$} & \tabincell{l}{FA:($2(n-1)N$)\\ BS:$1$\\FS:$(n-1)N$}& \tabincell{l}{nA:($(n-1)N$)\\ nS:$\alpha n N$} &\tabincell{l}{FA:$(\alpha n -1)N$\\RM:$N$} & \tabincell{l}{FA:$(p_{\chi}n-1)(p_{\chi}+\alpha)N$\\BS:$1$, aS:$1,$ \\FS: $(p_{\chi}n-1)(\alpha N)$} \\
        \hline
        Communication &$0$ & qd:$2nN+N$ in $2$ rounds & \tabincell{l}{nd:$\alpha n N(n+1)$\\ Nd:$(n+1)N$\\in $1$ round}&$0$ & \tabincell{l}{bd:$N$\\ Nd:$\alpha N(1+ p_\chi n)$\\
        qd: $\alpha N(1+p_\chi n)$\\in $3$ rounds}   \\
        \hline
    \end{tabular}
     \begin{tablenotes}
        \footnotesize
        \item FA: Field addition; RM: Real number multiplication; BS: Bernoulli sampling, aS: Sampling from $S_{N,\alpha N}$ following $\chi,$ FS: Field Sampling, nA: Addition of elements in $[n],$ nS: Sampling from $[n],$ nd: sending elements of $[n]$, Nd: sending elements of $[N],$ bd: sending of binary elements, qd: sending elements of $\mathbb{F}_q$
      \end{tablenotes}
          \caption{Complexity cost requirements for $\mathtt{DPDS}$ and $\mathtt{TSS}$}
    \label{comparison}
  \end{threeparttable}
\end{table*}

\noindent\textbf{Security against Collusion of Users:}

We note that since participating users secretly share their private values to a randomly chosen set of parties, if the size of such set is not sufficiently large, there is a non-zero probability that all the chosen parties collude, which enable them to recover all the reporting users' private values. In order to avoid such case, we can make the following minor adjustment to Algorithm~\ref{2SSP} to provide privacy against collusion of users. Suppose that there is a positive integer $\phi$ such that the maximum number of colluding parties is $\phi.$ Then we can provide privacy guarantee if for any item, there are at least $\phi+1$ elected users to assist in the reporting process to $\mathcal{S}_2.$ This can be done by making the following adjustments to Lines $13$ up to $20.$ 

\begin{algorithm}[!htbp]
\caption{Two-stage sampling ($\mathtt{TSS}'$)}\label{2SSPPrime}
\hspace*{0.02in} {\bf Input:}
Private values of $n$ users $\mathbf{t}_i\in\{\mathbf{e}_1,\cdots, \mathbf{e}_N\}\subseteq \mathbb{F}_q^N,$ sampling probability $p,$ reporting proportion $\alpha,$ reporting distribution $\chi\in\{\chi_U,\chi_A\}$ (with possible reporting parameter $\gamma$ if necessary), reporter bound $\phi,$ and privacy budget $\epsilon;$\\
\hspace*{0.02in} {\bf Output:}
Estimate of normalized frequency of each value $\widehat{\boldsymbol \psi};$
\begin{algorithmic}[1]
\STATE \COMMENT{Follow the first 11 lines of $\mathtt{TSS}$ in Algorithm~\ref{2SSP}}
\FOR{$j=1,\cdots, N$}
\STATE $\mathcal{S}_1$ sets $m_j^\prime=\max(\phi+1,m_j);$
\STATE $\mathcal{S}_1$ randomly selects $m_j^\prime$ users $u_{j,1},\cdots, u_{j,m_j};$
\STATE $\mathcal{S}_1$ publishes $(j,m_j,u_{j,1},\cdots, u_{j,m_j^\prime})$ to all users and the aggregating server $\mathcal{S}_2;$
\FOR{$i$ such that $j\in A_i$}
\STATE $u_i$ with encoded response for item $I_j, \widehat{t_{i,j}}$ samples $m^\prime_j-1$ random elements of $\mathbb{F}_q,$ denoted by $s_{j,i,1},\cdots,s_{j,i,m^\prime_j-1}\in \mathbb{F}_q$ and sets $s_{j,i,m^\prime_j}= \widehat{t_{i,j}}-\sum_{t=1}^{m^\prime_j-1} s_{j,i,t}\in \mathbb{F}_q;$
\STATE $u_i$ sends $(j,s_{j,i,t})$ to $u_{j,t}$ for $t=1,\cdots, m_j^\prime;$
\ENDFOR
\STATE \textit{[For elected users $u_{j,1},\cdots, u_{j,m^\prime_j}$]}
\STATE Upon receiving $m_j$ values $s_{j,i,t}$ for $i$ such that $j\in A_i,$ $u_{j,t}$ computes $s_{j,t}=\sum_{i:j\in A_i} s_{j,i,t}$ and sends $(j,s_{j,t})$ to $\mathcal{S}_2;$
\STATE \textit{[$\mathcal{S}_2$ Side]}
\STATE $\mathcal{S}_2$ calculates the estimate of the count for item $j$ as $s_j=\sum_{t=1}^{m_j^\prime} s_{j,t}$ and treat it as a real number;
\STATE $\mathcal{S}_2$ performs some post-processing to obtain the unbiased estimator for the count of item $I_j$ as $\widehat{\psi_j}=\frac{1}{q_\chi n}s_j$ where $q_\chi=p p_\chi$ and
\[p_\chi=\left\{
\begin{array}{cc}
    \alpha, &\mathrm{~if~}\chi=\chi_U  \\
     \frac{\alpha\gamma}{\alpha\gamma+1-\alpha},&\mathrm{~if~}\chi=\chi_A 
\end{array}
\right. .
\]

\ENDFOR
\end{algorithmic}
\end{algorithm}

It is easy to see that such minor adjustments do not affect the privacy against any of the servers nor the utility of $\mathtt{TSS}.$ However, it provides a privacy guarantee of any user's private value against any collusion of up to $\phi$ other users. This shows that $\mathtt{TSS}'$ provides privacy against an honest-but-curious adversary controlling either $\mathcal{S}_1,\mathcal{S}_2,$ or any $\phi$ out of the $n$ users. We note that here the privacy guarantee is no longer applicable if the adversary controls two servers or any one server along with the $\phi$ users. Furthermore, the adjustments done in $\mathtt{TSS}'$ affects the complexity for the participants. We note that for any $j$ such that $m_j<\phi+1,$ it incurs extra $\phi+1-m_j$ sampling over $[n]$ and publication of the value of $m_j\in [n]$ to $\mathcal{S}_1.$ Similarly, for any $j$ such that $m_j<\phi+1,$ it incurs an extra $\phi+1-m_j$ field addition operations to $\mathcal{S}_2.$ Lastly, for any $j$ such that $m_j<\phi+1,$ users that decide to report his $j$-th value get $\phi+1-m_j$ more field sampling, $\phi+1-m_j$ more field additions and $\phi+1-m_j$ field elements to be sent. For users that are elected to assist the reporting process of item $j,$ they will need additional $\phi+1-m_j$ field additions.

\subsection{Weighted Aggregation for hybrid privacy level}
Having one level of privacy guarantee for all the users of a mechanism may provide a simpler mechanism to analyze. However, since different users may have different privacy preferences, such practice may sacrifice statistical accuracy by providing a higher level of privacy guarantee than what are required by users. A natural way to improve the statistical accuracy is to have few levels of privacy settings, such as \textit{Very Strong}, \textit{Strong}, \textit{Normal}, and \textit{Weak} for the users to choose depending on their respective privacy preferences. Such solution is not necessarily new. There have been some studies on the effect of having multiple levels of privacy on the statistical accuracy of the mechanism \cite{akter2017computing,ye2019multiple}. However, in such studies, aggregation is done at once regardless of the choice of the level of distortion for each users' report. Although such aggregation technique yields a good level of statistical accuracy, it is not optimal in some cases, for example, when we assume that the private data comes from a fixed distribution. Intuitively, the statistical accuracy of the aggregation can be improved further if reports with smaller levels of distortion have larger weights in the aggregation. In order to formalize such aggregation method, we propose a general theoretical framework for data aggregation with different levels of perturbation.

Here we assume that the users' data are independently and identically distributed random variables with expected value equal to the desired average $\bar{v}.$ This assumption applies in scenario where the users' data follow a fixed distribution, such as the human weight data that follows a Gaussian distribution \cite{a2009height},
 web page click frequency, and node degrees in social networks which follow a power-law distribution \cite{clauset2009power,gong2012evolution}.
We further assume that the output of each user is an unbiased estimator of his true value, which can be achieved by using any unbiased perturbation mechanism. 
The general idea is that we assign different weights to each group to calibrate the statistical mean. 
We assume that there are $\mu$ different privacy levels $\{\epsilon_1,\cdots, \epsilon_\mu\}$ where the $j$-th privacy level is chosen by $n_j$ users. For any $j=1,\cdots, \mu$ and $k=1,\cdots, n_j,$ we denote the $k$-th user choosing the $j$-th privacy level as $u_{j,k}.$ Suppose that we assign the weight $w_j\in \mathbb{R}_{>0}$ to the $j$-th privacy level. If the output of $u_{j,k}$ is $\tilde{v}_{j,k},$ the average $\bar{v}$ is estimated as $
 \widehat{\overline{v}}=\frac{\sum_{j=1}^\mu w_j \sum_{k=1}^{n_j} \tilde{v}_{j,k}}{\sum_{j=1}^\mu n_j w_j}.
$

We determine the values of $w_j$ that provides high accuracy, which highly depends on the variance of each estimation $\tilde{v}_{j,k}.$ For simplicity, we assume that the variance of $\tilde{v}_{j,k}$ is the same for users with the same privacy level $\epsilon_j$ and let such variance be denoted by $V_j.$ Such assumption is applicable for the majority of the existing DP mechanisms.

We expect to get a more accurate estimation by assigning more weights to users who choose a larger privacy budget. We evaluate the statistical error by calculating its expected squared error, $\Delta=\mathbb{E}\left[\left(\frac{\sum_{j=1}^\mu w_j \sum_{k=1}^{n_j} \tilde{v}_{j,k}}{\sum_{j=1}^\mu n_j w_j}-\bar{v}\right)^2\right]$ which can be computed as 
 \begin{equation}\label{eq:E=0} 
 \Delta=\mathrm{Var}\left[\frac{\sum_{j=1}^\mu w_j \sum_{k=1}^{n_j} \tilde{v}_{j,k}}{\sum_{j=1}^\mu n_j w_j}\right]
\end{equation}
or $\Delta= \frac{1}{(\sum_{j=1}^{\mu} n_j w_j)^2}\sum_{j=1}^{\mu} w_j^2 n_j V_j.$
Here Equation~\eqref{eq:E=0} comes from the fact that $\tilde{v}_{j,k}$ is an unbiased estimator of $v_{j,k}$ which has an expected value of $\bar{v}.$ 
 
To correct the statistical error, we propose two solutions to get a highly accurate statistical mean. 

\noindent\textit{Solution 1.} We model the weight assignment process as an optimization problem, 
\[\min\left\{\frac{\sum_{j=1}^\mu n_i w_i^2V_j}{(\sum_i^\mu n_i w_i)^2}: \sum_i^{\mu} w_i=1, 0<w_i<1, \forall i\in [\mu]\right\}\]
Then we can numerically solve the optimization problem using any existing optimization algorithm such as gradient descent.

\noindent\textit{Solution 2.} Treating the expected squared error formula as a multi-variable function of the weights, we may consider its critical points for possible local minima. One such critical points is $(w_1,\cdots, w_\mu)$ where, for $i=1,\cdots, \mu,$
\begin{equation}\label{eq:goodweight}
    w_i=\frac{1/V_i}{\sum_{j=1}^{\mu}\frac{1}{V_j}}.
\end{equation}

\begin{lemma}\label{lem:minw}
Define a function $f:\mathbb{R}^{\mu}\rightarrow \mathbb{R}$ such that $f(w_1,\cdots, w_\mu)=\frac{1}{(\sum_{j=1}^{\mu} n_j w_j)^2}\sum_{j=1}^{\mu} w_j^2 n_j V_j.$
The critical points of $f$ are of the form $c\cdot\left(\frac{1}{V_1},\cdots,\frac{1}{V_\mu}\right)$ for $c\in \mathbb{R}.$
\end{lemma}
\begin{proof}
For $t=1,\cdots, \mu,$ define $f_t\triangleq\frac{\partial f}{\partial w_t},$ the partial derivative of $f$ with respect to $w_t.$ Then
\[
    f_t=\frac{2n_t}{(\sum_{j=1}^\mu w_jn_j)^2}\left[w_t V_t-\frac{\sum_{j=1}^{\mu} w_j^2 n_j V_j}{\sum_{j=1}^\mu n_j w_j}\right].
\]
So letting $f_t=0$ and subtracting $f_{t'}$ from $f_t$ for any $t'\neq t,$ we get 
\begin{equation}\label{eq:ft-tprime}
    (w_t V_t- w_{t'} V_{t'})\left(\sum_{j=1}^{\mu}w_j n_j\right)=0.
\end{equation}
Since $\sum_{j=1}^\mu w_j n_j >0,$ Equation~\eqref{eq:ft-tprime} implies that $w_t V_t= w_{t'}V_{t'}$ for any $t\neq t'.$ This shows that $w_t V_t$ is constant or any $t.$ Let such constant be $c.$ This shows that $f_t=0$ for any $t$ if and only if there exists $c$ such that $w_t=\frac{c}{V_t}.$ So the critical points of $f$ are in the form $c\cdot\left(\frac{1}{V_1},\cdots,\frac{1}{V_\mu}\right).$ 
\end{proof}

Note that although the given $(w_1,\cdots, w_\mu)$ in Eq.~\eqref{eq:goodweight} is a critical point of the function $f,$ it only tells us that the point $(w_1,\cdots, w_\mu)$ may be a local minimum of the function. However, this does not guarantee that using such weights provides a smaller expected error compared to an unweighted estimator, i.e. $w_i=1$ for all $i=1,\cdots, \mu.$ The following lemma confirms that when we set $w_i=\frac{1/V_i}{\sum_{j=1}^\mu 1/V_j},$ the expected error is indeed smaller than the expected error when we set $w_i=1$ for all $i=1,\cdots,\mu.$

\begin{lemma}\label{lem:minwbetter}
For $i=1,\cdots, \mu,$ let $w_i$ be as defined in Eq.~\eqref{eq:goodweight}. Then $f(w_1,\cdots, w_\mu)\leq f(1,\cdots,1).$ Furthermore, equality holds if and only if $V_i=V_j$ for any $i,j\in\{1,\cdots, \mu\}.$
\end{lemma}
\begin{proof}
Let $E_1=f(w_1,\cdots, w_\mu)$ and $E_2=f(1,\cdots,1).$ It is easy to verify that $E_1=\frac{1}{\sum_{j=1}^\mu\frac{n_j}{V_j}}$ and $E_2 = \frac{\sum_{j=1}^\mu n_j V_j}{\left(\sum_{j=1}^{\mu} n_j\right)^2}.$ Then $E_1\leq E_2$ if and only if
\begin{equation}\label{eq:WTSminwbetter}
\sum_{1\leq i<j\leq \mu} 2n_in_j \leq \sum_{1\leq i<j\leq \mu} n_i n_j \left(\frac{V_i}{V_j}+\frac{V_j}{V_i}\right).
\end{equation}
Note that for any $i=1,\cdots, \mu$ and $j=i+1,\cdots,\mu,$ since $(V_i-V_j)^2\geq 0$ and it is $0$ if and only if $V_i=V_j,$ we have $\frac{V_i}{V_j} + \frac{V_j}{V_i} \geq 2$ for any positive real numbers $V_i$ and $V_j$ and equality is achieved if and only if $V_i=V_j.$ This concludes the proof.
\end{proof}

Lemma~\ref{lem:minwbetter} shows that unless we have the same variance for all the privacy levels, using weighted average with weights as defined in Eq.~\eqref{eq:goodweight} provides a statistically more accurate estimate compared to the unweighted average. Note that since we assume the estimator is unbiased, having smaller mean squared error implies that we also have an estimator with smaller variance, which again shows that the weighted average provides a better statistical accuracy.
\\
\textbf{Remark.} Note that in addition to the guarantee that weighted aggregation achieves better statistical accuracy compared to unweighted aggregation, it is also done as a post-processing calculation which consumes no extra privacy budget. We also note that such post-processing technique only requires an additional of $\mu$ real number multiplication operations in the multiplication of aggregated data for different privacy settings by its corresponding weights, which is asymptotically negligible, especially when $\mu$ is assumed to be small.

\section{Experiment} \label{expe}
\subsection{Experiment Setup}
\subsubsection{Dataset} We test the proposed methods on both real and synthetic datasets.
\begin{itemize}
    \item \textit{Gowalla.} Gowalla dataset is collected from a location-based social network. It includes millions of check-in information from Nov 2009 to Dec 2011. We extract the records with location in the range of $[30,45]\times[-100,-80]$ on the map and partition it into cells of $5 \times 5.$ Each user then merges his data records to one indicating his mostly frequently visited area within the cells. 
    \item \textit{Census-Income.} Census-Income dataset contains census data extracted from the 1994 and 1995 Current Population Surveys conducted by the U.S.Census Bureau. We extract the income attribute and group the data to subgroups, each corresponding to disjoint intervals of length $100.$  
    \item \textit{Synthetic.} We generate a random data set with $1000$ population, each user holding $1$ of $30$ possible items which is sampled uniformly at random.
\end{itemize}
\subsubsection{Metrics} We evaluate the accuracy of the estimation methods by a traditional metric, mean squared error (MSE), which is defined as $MSE = \frac{1}{N}\sum_{i=1}^N(\psi_i -\widehat{\psi_i})^2.$
\subsubsection{Environment}  
All algorithms are implemented in MATLAB and tested on a remote server with 62.87GB RAM. We run each algorithm 20 times and report the average result. 

\subsection{Performance Evaluation}
\subsubsection{Effect of Privacy Level to $\mathtt{DPDS}$}
Given the uniformity assumption $\beta,$ we randomly sample $1000$ data records from each processed data set under the assumption that each item occurs at least $1000\beta$ times. Along with the privacy budget $\epsilon$ and the number of different items $N,$ we determine the failure probability $\delta$ that satisfies \eqref{bound}.
To show the advantage of sampling, we compare $\mathtt{DPDS}$ against $\mathtt{DPDG}$ for various privacy levels $\epsilon\in\{0.1,0.2,\cdots, 1\}.$ 

As shown in Fig. \ref{fig1}, $\mathtt{DPDS}$ performs much better than the $\mathtt{DPDG}$ for the same privacy level over all three data sets. Furthermore, it achieves much higher accuracy when the privacy budget is small. For example, $\mathtt{DPDS}$ has MSE of around $0.01$ when $\epsilon = 0.1$ on Gowalla data set, which is an improvement of over $90\%$ compared to the Gaussian mechanism when $\delta = 10^{-7}$. It even performs better than Gaussian mechanism with much weaker privacy guarantee ($\delta = 0.5$) when $\epsilon<0.6$. Similar results can be found over Census-Income data set and Synthetic data set. Such observation may be explained from the fact that $\mathtt{DPDS}$ uses sampling for the perturbation, which comes with a much smaller error compared to that of Gaussian mechanism that injects Gaussian noise. Gaussian mechanism has a large variance in scenarios with small privacy budget and small population size.

\begin{figure*}[ht]
     \centering
\subfigure[Gowalla]{
\label{Fig-Gowalla}
\includegraphics[scale=0.3]{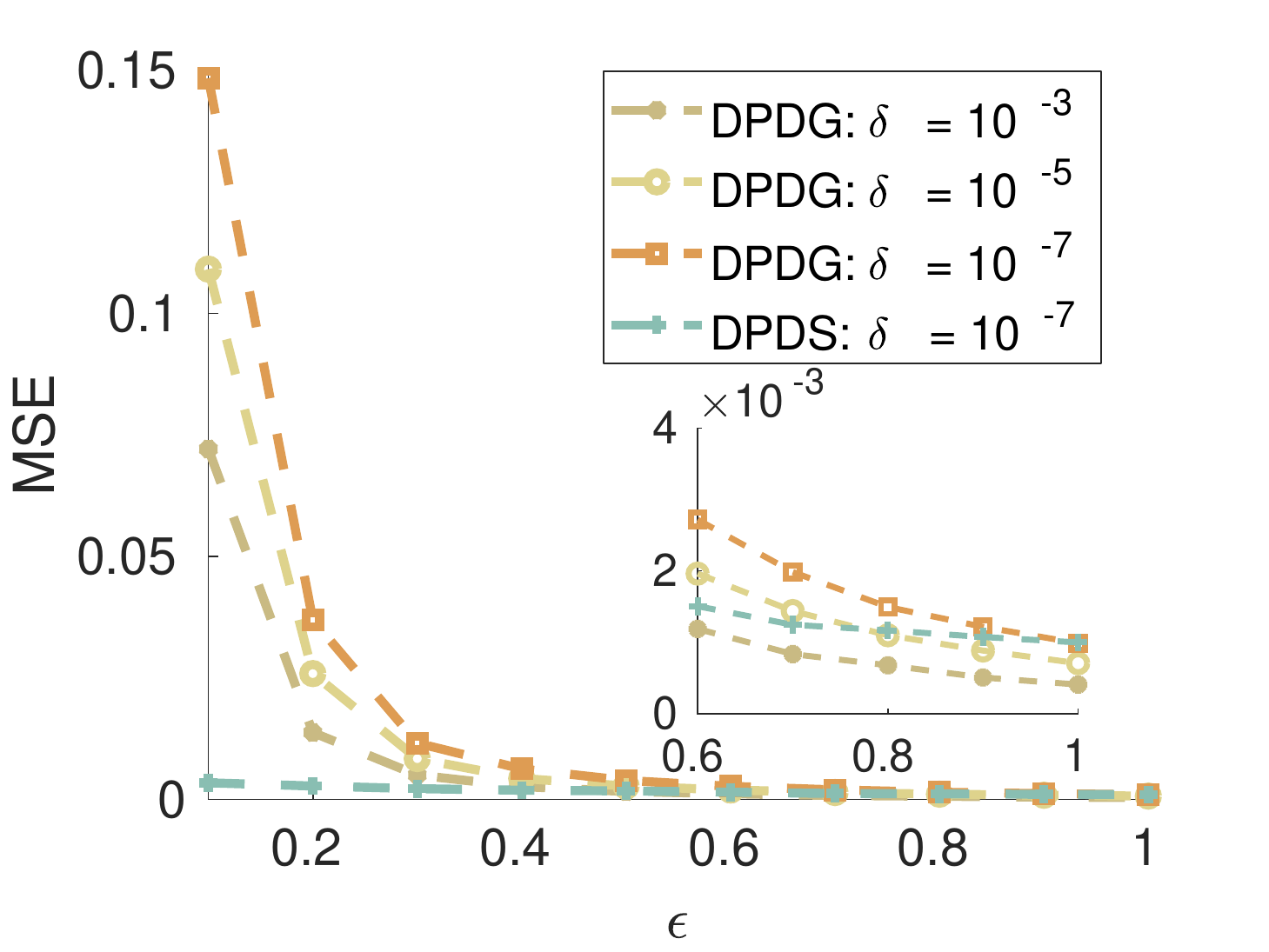}
}
\subfigure[Census-Income]{
\label{Fig-USIncome}
\includegraphics[scale=0.3]{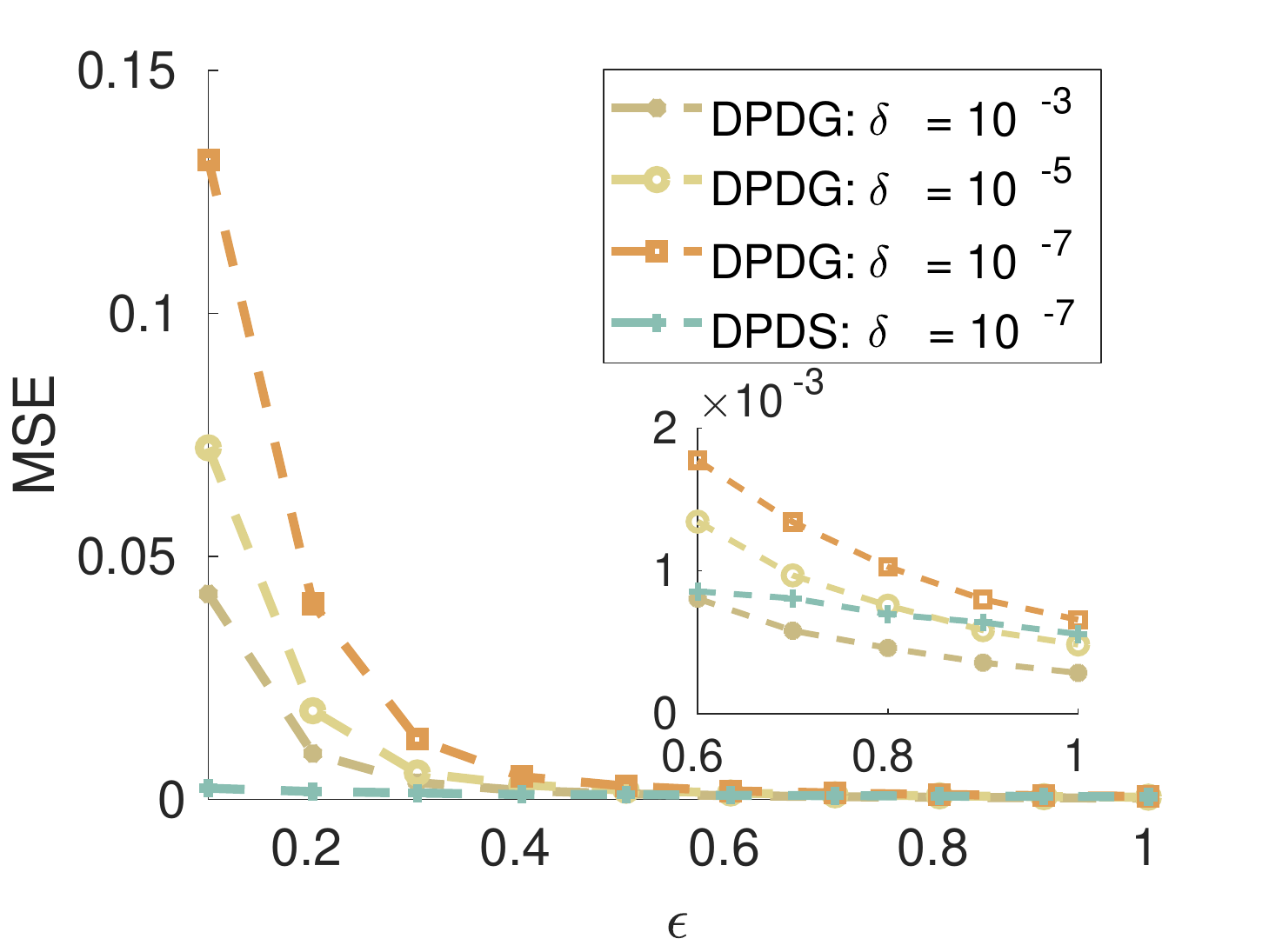}
}
\subfigure[Synthetic]{
\label{Fig-Synthetic}
\includegraphics[scale=0.3]{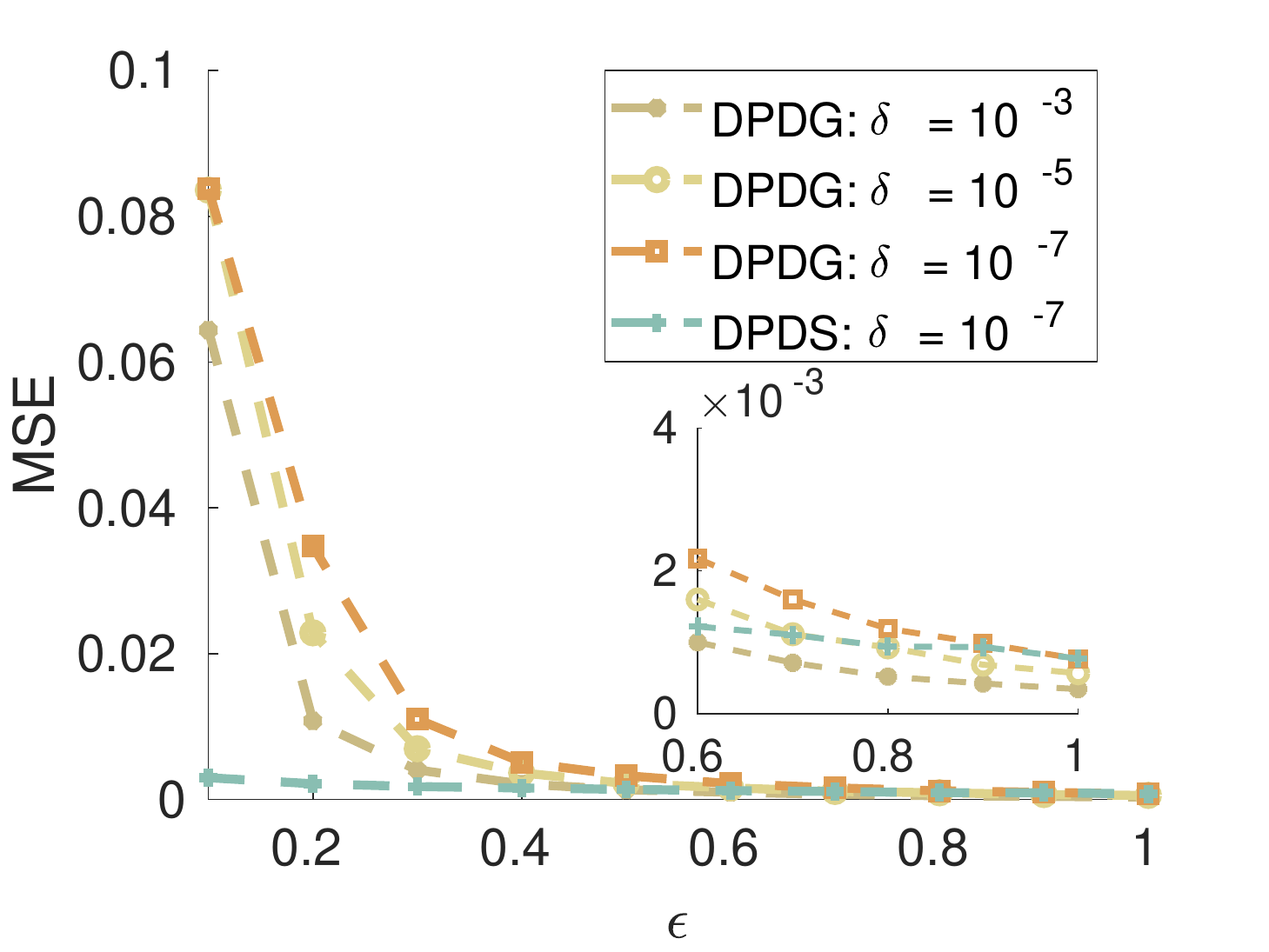}
}
\caption{Performance comparison of $\mathtt{DPDS}$ and $\mathtt{DPDG}$ for various $\epsilon$ and $\delta$}
\label{fig1}
\end{figure*}

\subsubsection{Effect of Population Size on $\mathtt{DPDS}$}
Next, we consider the effect that the number of participants gives to the performance of the mechanism $\mathtt{DPDS}.$ To simulate this, we sample $n$ data each from the three data sets where $n\in\{1000,2000,\cdots, 5000\}.$ We also vary the privacy budget $\epsilon\in\{0.2,0.4,\cdots, 1\}.$

As can be observed in Fig. \ref{fig2}, there is an obvious accuracy improvement when the size of participants increases over all data sets. 
For example, when Gowalla data set is considered, setting $\epsilon = 0.2, \mathtt{DPDS}$ has a mean squared error of around $2.4\times 10^{-3}$ when there are $1000$ participants while it achieves an MSE of around $0.9\times 10^{-3}$ for $5000$ participants. We find a similar trend over both Census-Income data set and Synthetic data set. 
The experiment results are consistent with the accuracy analysis that the statistical error is inversely proportional to the population size. 

\begin{figure*}[ht]
     \centering
\subfigure[Gowalla]{
\label{Fig2-Gowalla}
\includegraphics[scale=0.3]{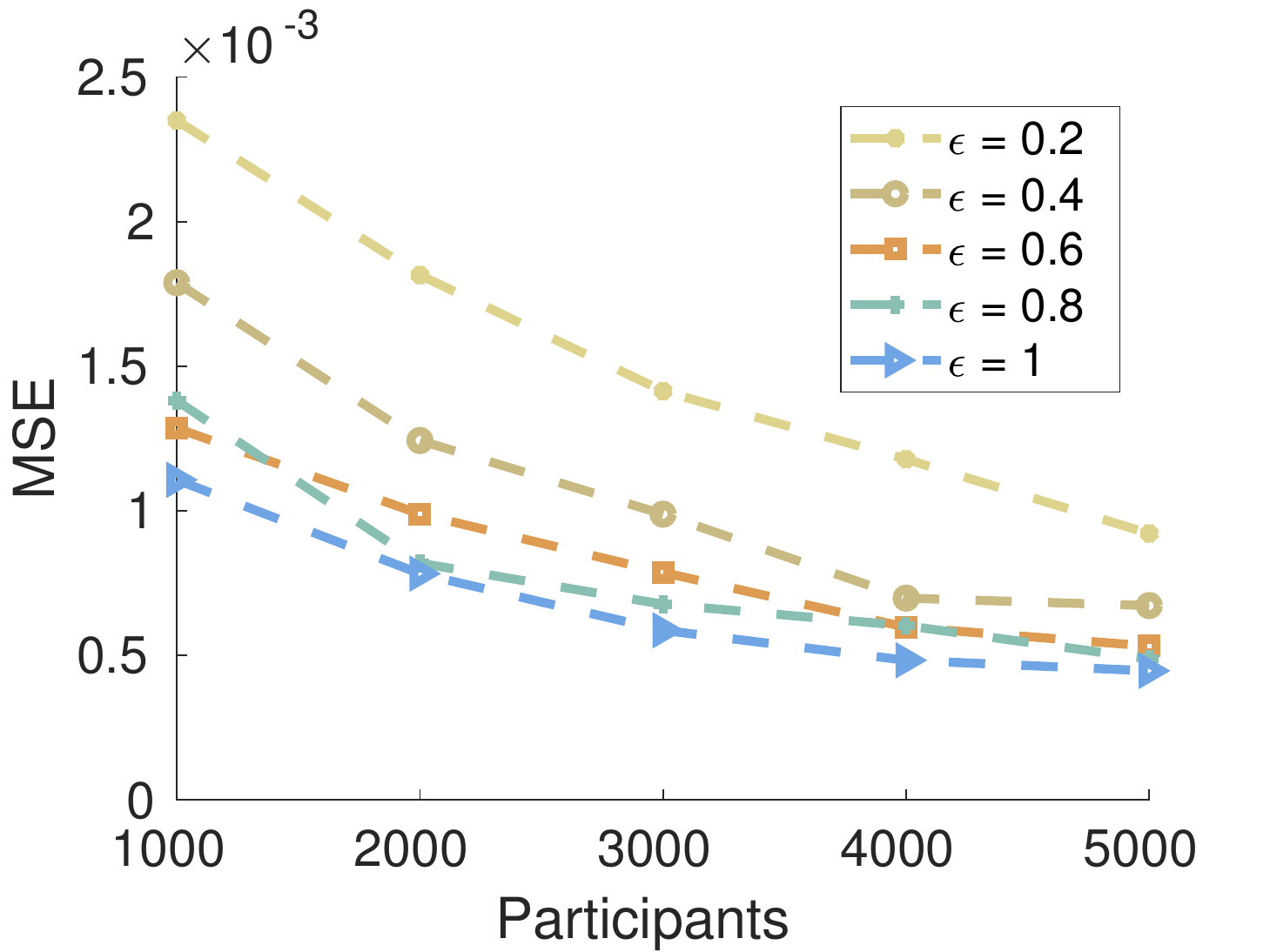}
}
\subfigure[Census-Income]{
\label{Fig2-USIncome}
\includegraphics[scale=0.3]{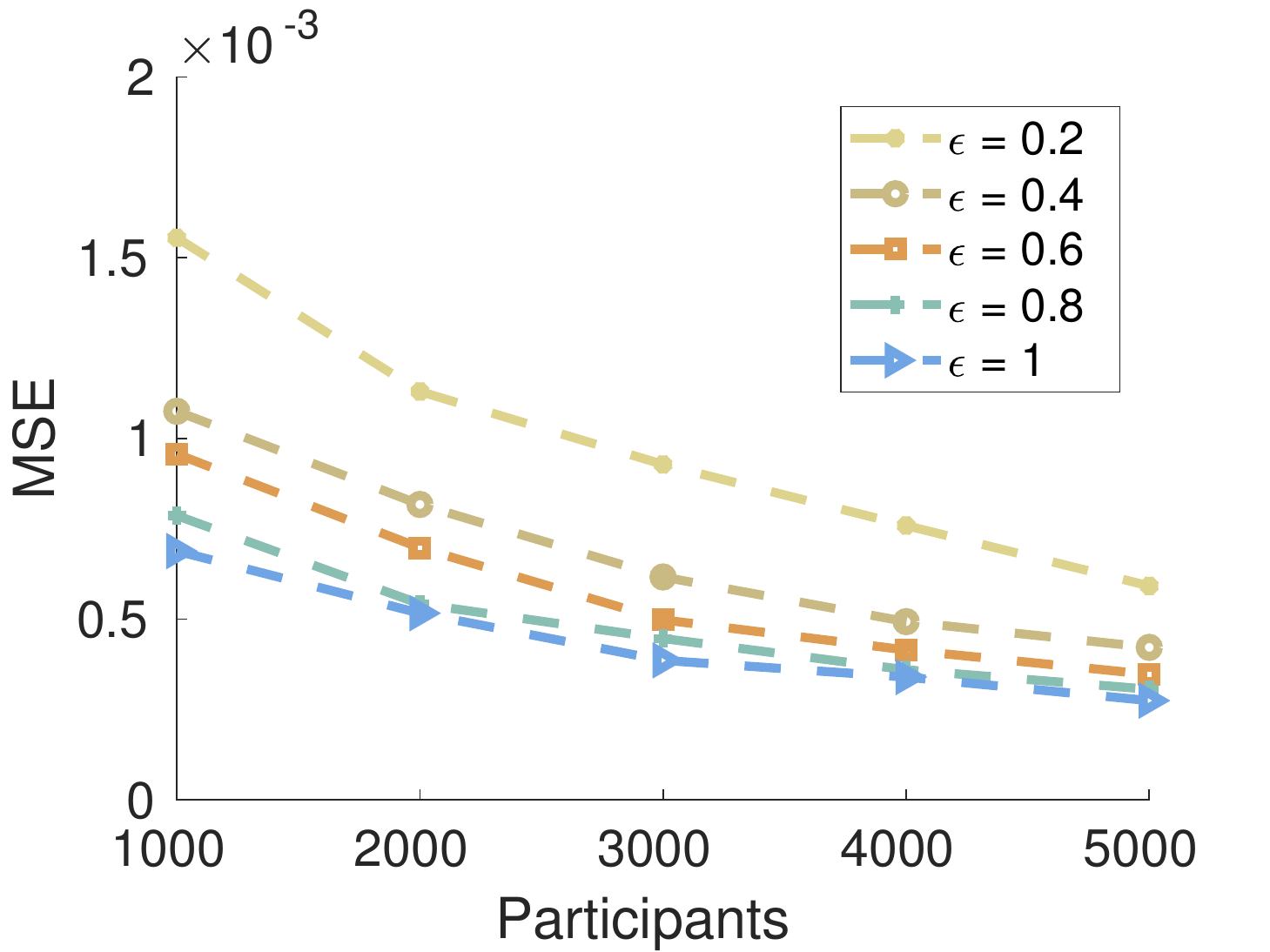}
}
\subfigure[Synthetic]{
\label{Fig2-Synthetic}
\includegraphics[scale=0.3]{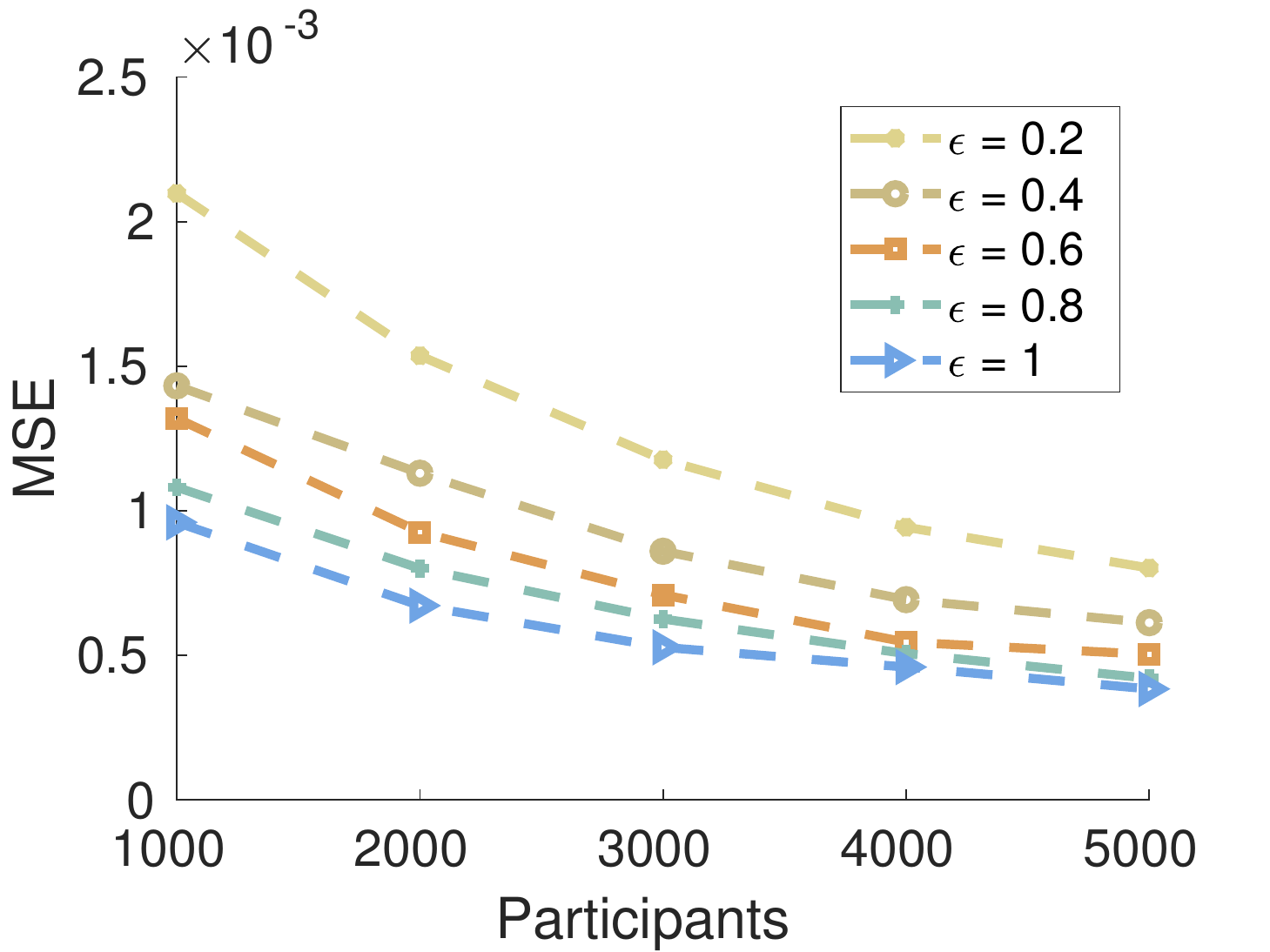}
}
\caption{Performance comparison of $\mathtt{DPDS}$ under different population sizes}
\label{fig2}
\end{figure*}

\subsubsection{Performance of Two-Stage Sampling}
To compare the effect of different ways of sampling on top of the first stage sampling, we consider $\alpha = 0.4$ and compare the uniform sampling $(US)$ to adaptive sampling $(AS)$ by varying the sampling probability over three data sets. 

As shown in Fig. \ref{fig3}, the adaptive sampling achieves a smaller MSE compared to uniform sampling under all settings. 
More specifically, having larger $\epsilon^\ast$ provides the adaptive sampling with a more statistically accurate estimate. This is because for the same sampling probability, the adaptive sampling lets an actively participating user to have a higher chance to report his true value. The larger $\epsilon^*,$ the higher the probability that the true value included in the report values. 

\begin{figure*}[ht]
     \centering
\subfigure[Gowalla]{
\label{Fig3-Gowalla}
\includegraphics[scale=0.3]{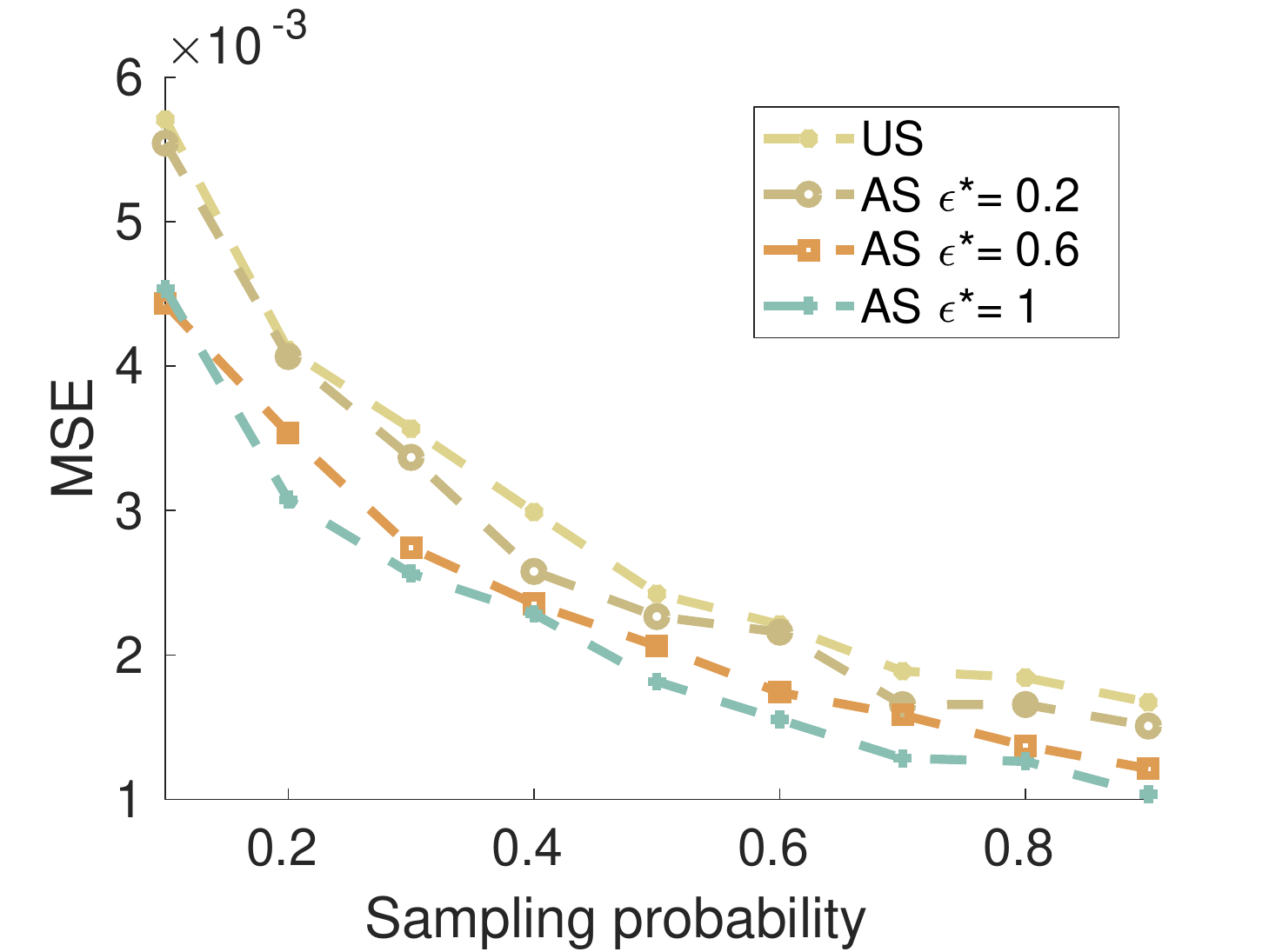}
}
\subfigure[Census-Income]{
\label{Fig3-USIncome}
\includegraphics[scale=0.3]{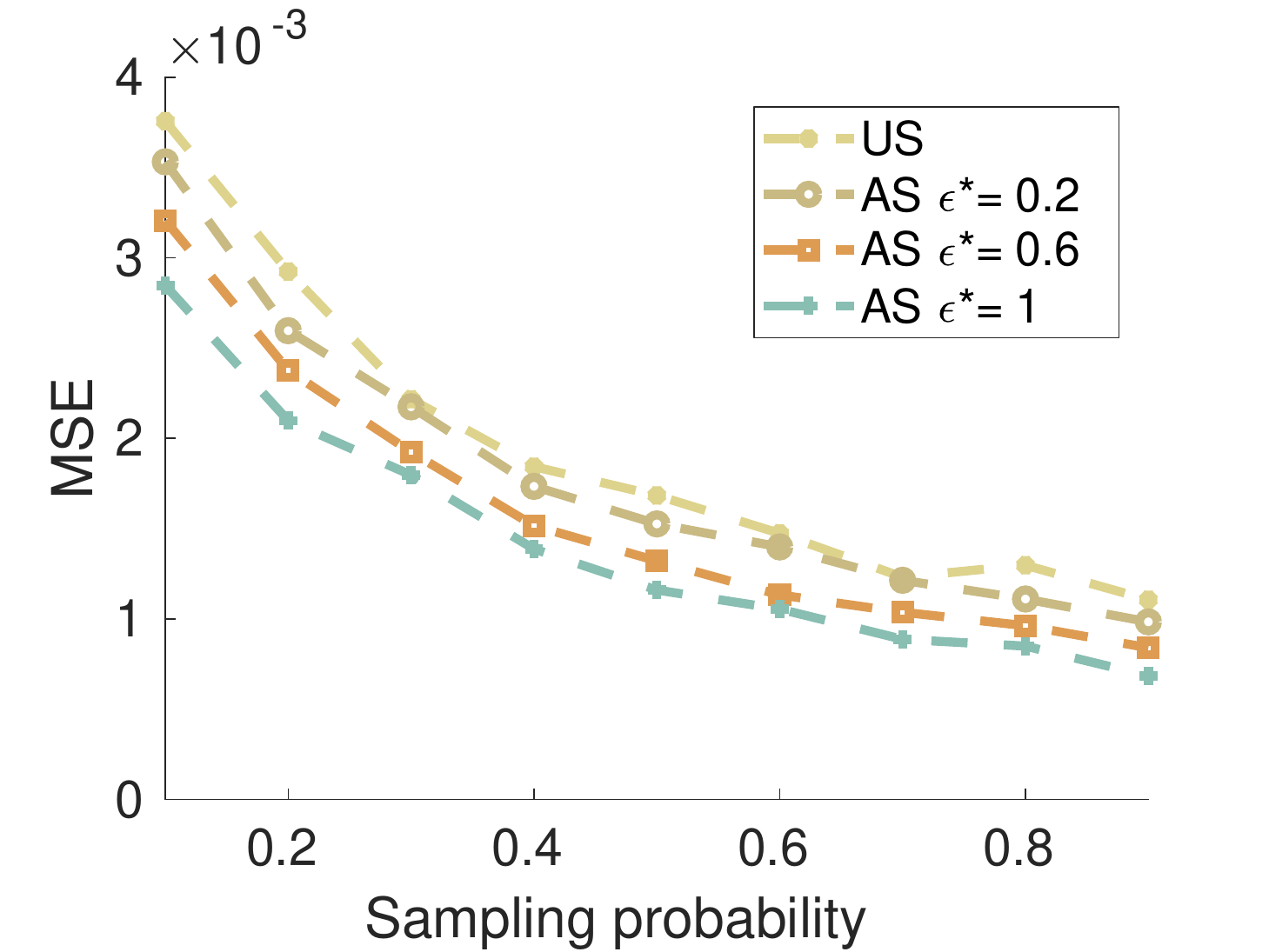}
}
\subfigure[Synthetic]{
\label{Fig3-Synthetic}
\includegraphics[scale=0.3]{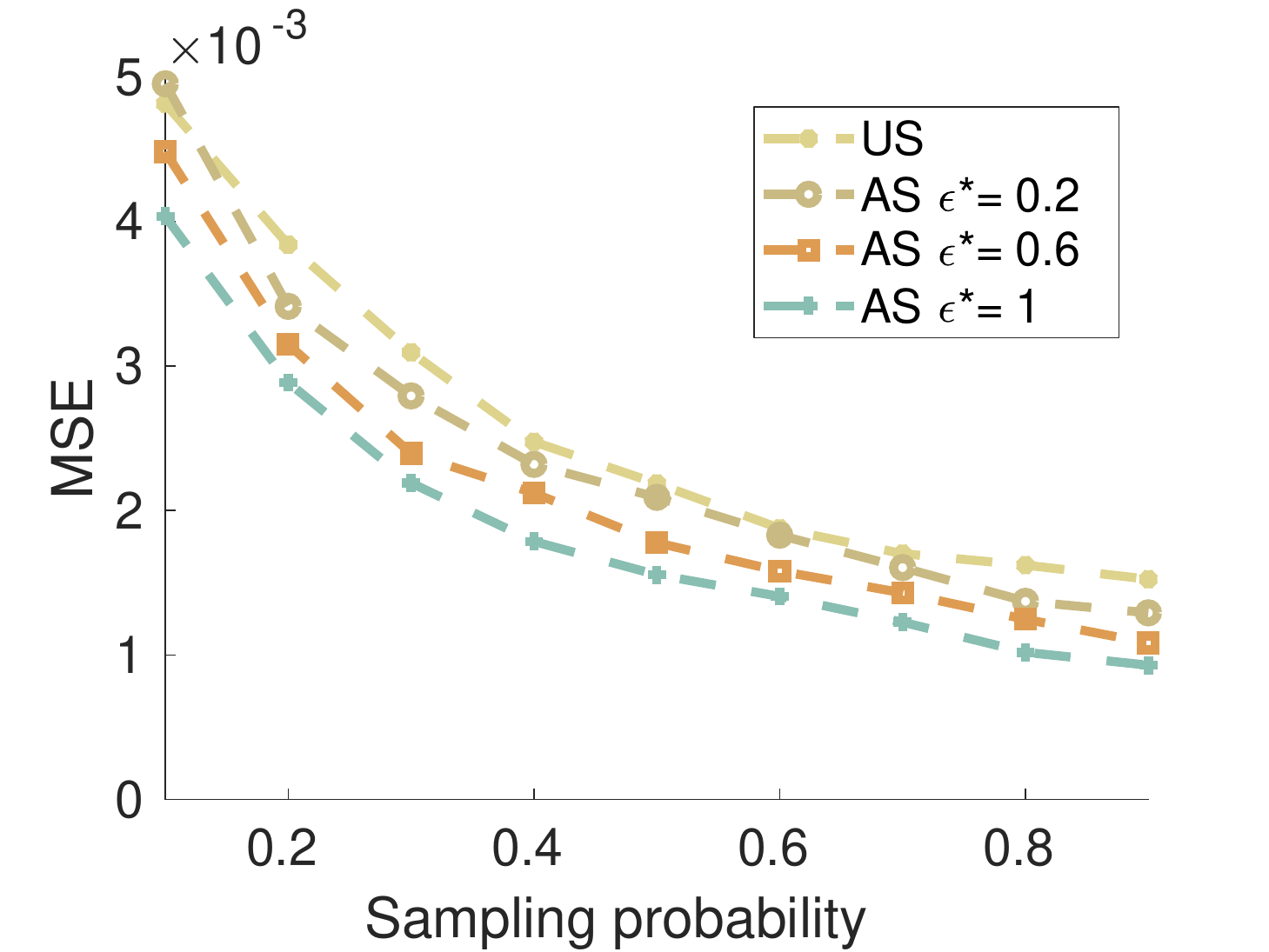}
}
\caption{Performance of two-stage sampling}
\label{fig3}
\end{figure*}

\subsubsection{Performance of Weighted Aggregation}

\textit{Weighted aggregation for frequency estimation by sampling.} To show the performance of the weighted aggregation, we partition $1000$ users into four groups with the same size where each group is assigned different privacy budgets. The settings of privacy budget assignment is shown in the first column of Table \ref{weightedfe}. 
Settings $s_2$ and $s_3$ are used to simulate the case where the groups size are not the same.

We compare the weighted aggregation method with both unweighted method (denoted by UWA) and the method that is using the smallest privacy budget (denoted by CPA). Fig. \ref{fig4} shows that the weighted method has the smallest mean squared error compared to the other two methods over all settings. That is because the weighted aggregation reduces the overall statistical variance. 
To further examine the effectiveness of the proposed weighted aggregation framework, we test it on two types of data distributions, uniform distribution and standard normal distribution, both of which having $1000$ participants and $30$ different items. Table \ref{weightedfe} shows that the assigned weights for the four groups are the same for both data distributions with fixed privacy setting. We also observe that when considered on data with normal distribution, the mechanism has a relatively larger error compared to the implementation on data generated following the uniform distribution.

\begin{figure*}[ht]
     \centering
\subfigure[Gowalla]{
\label{Fig4-Gowalla}
\includegraphics[scale=0.3]{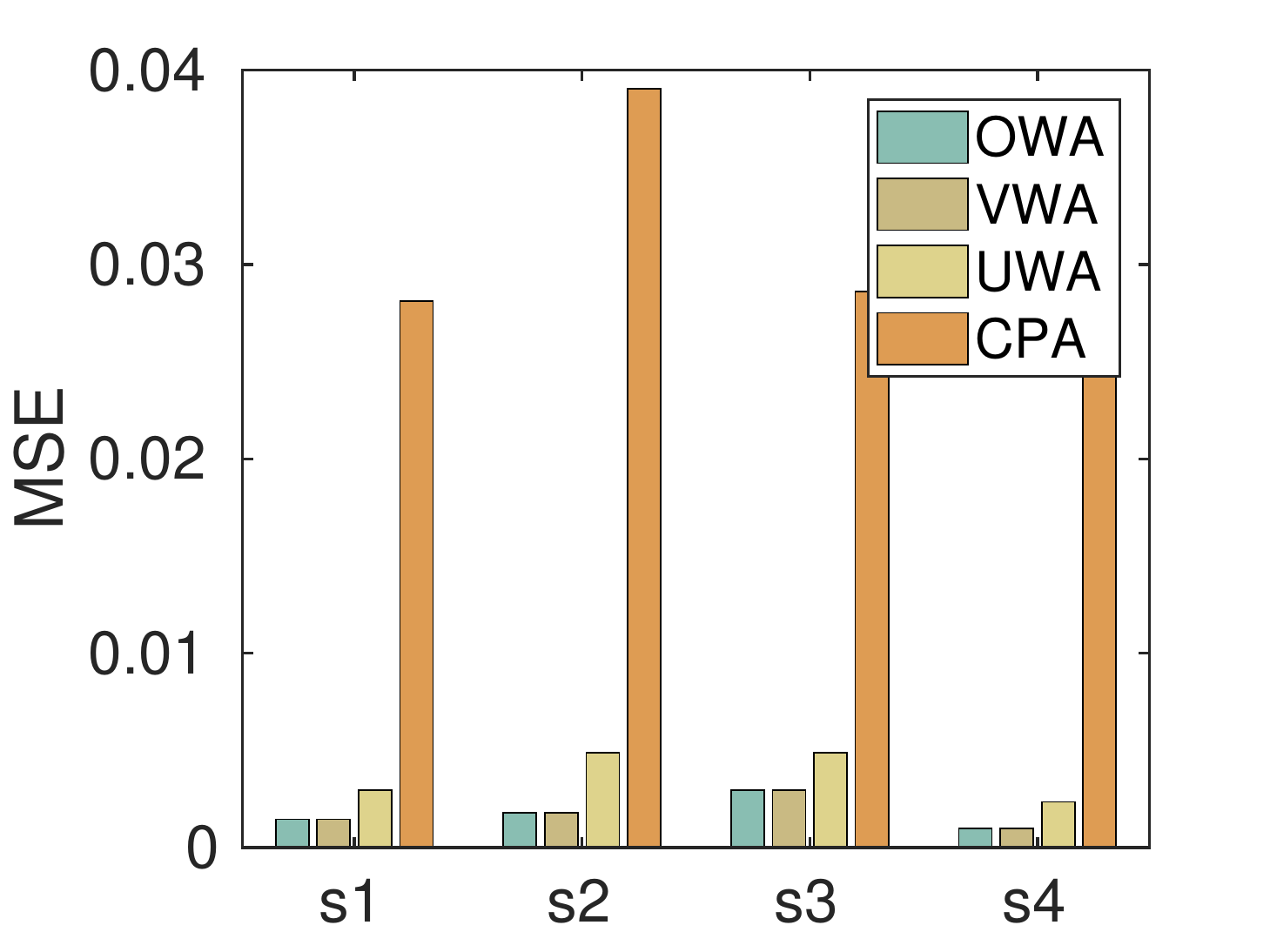}
}
\subfigure[Census-Income]{
\label{Fig4-Census-Income}
\includegraphics[scale=0.3]{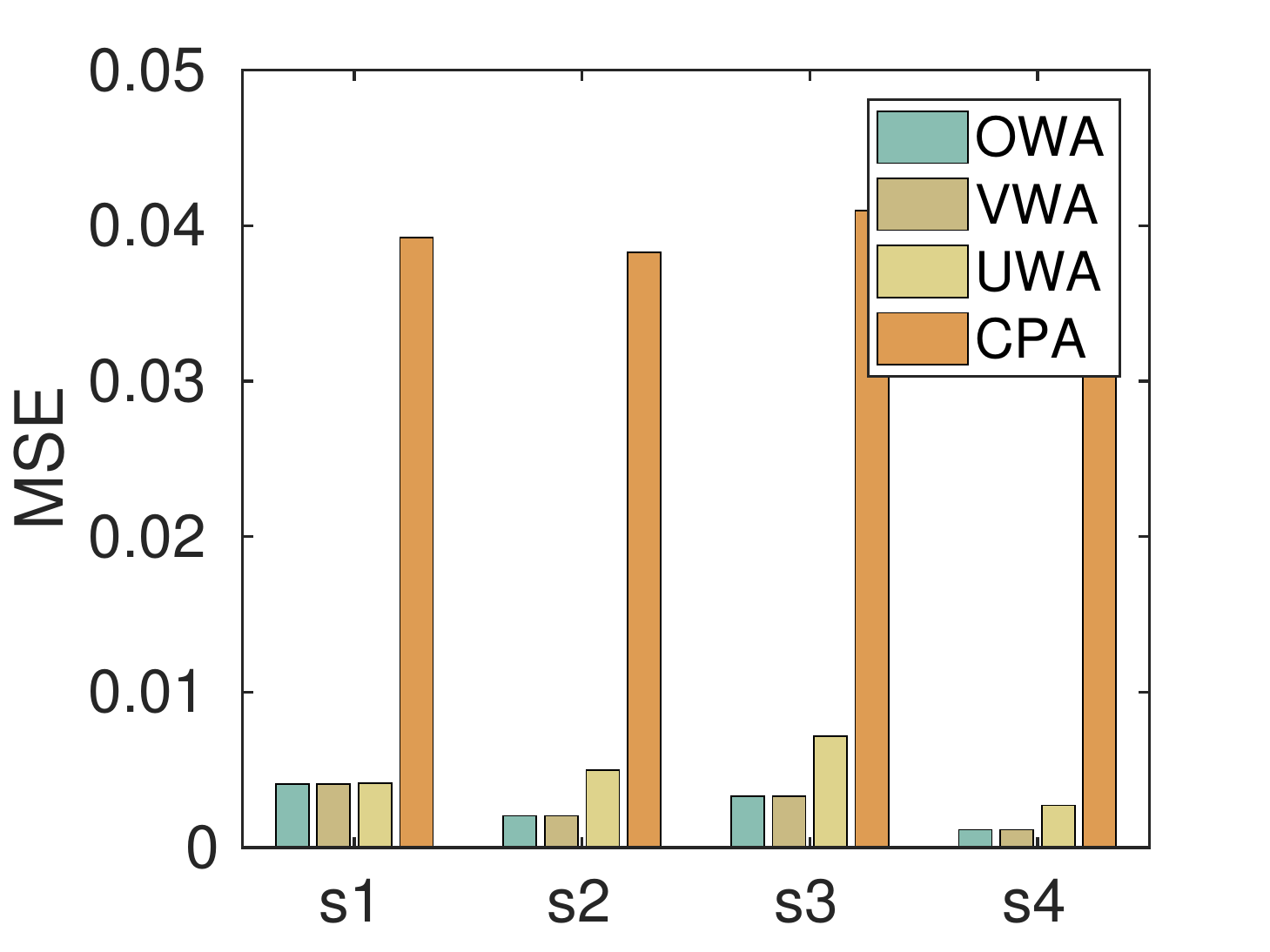}
}
\subfigure[Synthetic]{
\label{Fig4-Synthetic}
\includegraphics[scale=0.3]{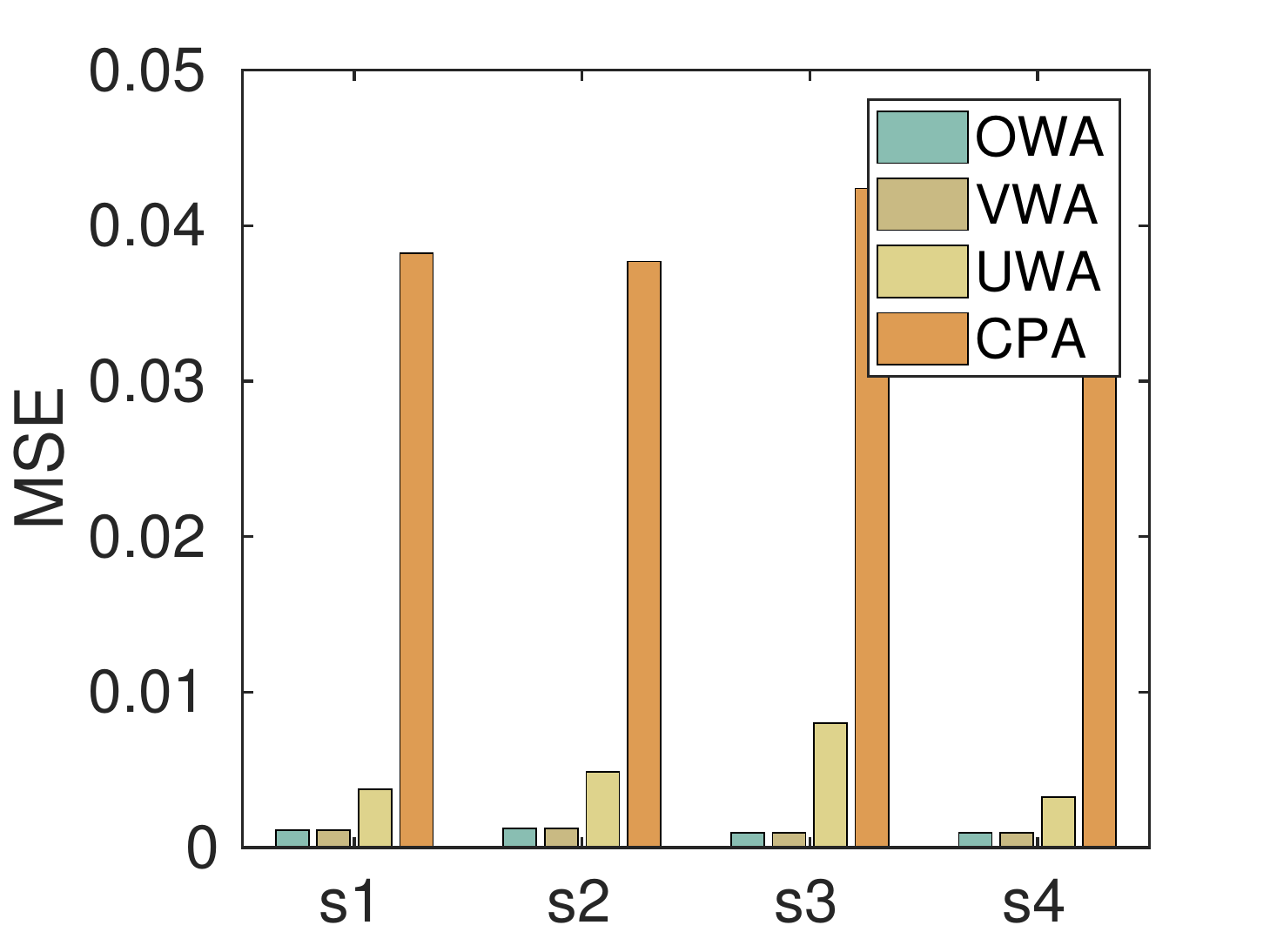}
}
\caption{Performance of weighted aggregation for frequency estimation}
\label{fig4}
\end{figure*}

\begin{table*}[ht]

\centering
\begin{threeparttable}[b]
    \begin{tabular}{|c|c|ccc|}
    \hline
    Privacy settings & Distribution & Method & Weights & MSE \\
    \hline
   \multirow{4}{*}{$s_1~ [0.1,0.4,0.7,1]$}& \multirow{2}{*}{Uniform} & OWA/VWA & $[0.0316,0.1477,0.3045,0.5162]$ & $1.124\times 10^{-3}$\\ 
   \cline{3-5}& & UWA & $[----]$ & $3.757\times 10^{-3}$ \\ 
   \cline{2-5} & \multirow{2}{*}{Normal} & OWA/VWA & $[0.0316, 0.1477, 0.3045, 0.5162]$ & $4.96\times 10^{-3}$\\ 
      \cline{3-5} & & UWA & $[----]$ & $5.72\times 10^{-3}$ \\ 
   \hline
   \multirow{4}{*}{$s_2~ [0.1, 0.1, 0.8, 1]$} & \multirow{2}{*}{Uniform} &OWA/VWA & $[0.0333, 0.0333, 0.3885, 0.5448]$ & $1.253\times 10^{-3}$ \\
   \cline{3-5}& & UWA & $[----]$ &$4.885\times 10^{-3}$ \\ 
      \cline{2-5} & \multirow{2}{*}{Normal} & OWA/VWA & $[0.0333, 0.0333, 0.3885, 0.5448]$ & $6.26\times 10^{-3}$\\ 
      \cline{3-5} & & UWA & $[----]$ & $6.51\times 10^{-3}$ \\ 
    \hline
 \multirow{4}{*}{$s_3~ [0.1, 0.1, 0.1, 1]$} & \multirow{2}{*}{Uniform}&OWA/VWA & $[0.0517, 0.0517, 0.0517, 0.8449]$ & $9.52\times 10^{-4} $\\ 
   \cline{3-5}& & UWA & $[----]$ & $7.999\times 10^{-3} $\\ 
   \cline{2-5} & \multirow{2}{*}{Normal} & OWA/VWA & $[0.0517, 0.0517, 0.0517, 0.8449]$ & $8.66\times 10^{-3}$\\ 
      \cline{3-5} & & UWA & $[----]$ & $1.4\times 10^{-2}$ \\ 
\hline
 \multirow{4}{*}{$s_4 ~[0.1, 0.8, 0.7, 1]$} & \multirow{2}{*}{Uniform}&OWA/VWA & $[0.0259,	0.3017, 0.2495, 0.4229]$ & $9.61\times 10^{-4} $\\ 
   \cline{3-5}& & UWA & $[----]$ & $3.233\times 10^{-3} $\\ 
   \cline{2-5} & \multirow{2}{*}{Normal} & OWA/VWA & $[0.0259,	0.3017, 0.2495, 0.4229]$ & $4.02\times 10^{-3}$\\ 
      \cline{3-5} & & UWA & $[----]$ & $4.17\times 10^{-3}$ \\ 
\hline
    \end{tabular}
  \begin{tablenotes}
  \item[*] OWA: weighted aggregation through optimization function VWA: weighted aggregation through proportion of variance UWA: unweighted aggregation CPA: consistent privacy aggregation with $\epsilon=0.1$. 
\end{tablenotes}  
\end{threeparttable}
\caption{Weights assignment for groups with different privacy settings}
\label{weightedfe}
\end{table*}

\textit{Weighted aggregation for other mechanisms.}
As discussed before, the proposed weighted aggregation does not only apply to the proposed sampling method. It also works for various perturbation methods and statistical analysis. To show the effectiveness of weighted aggregation over other mechanisms, we provide the experimental result on the use of the weighted aggregation in mechanisms using other frequency estimation methods. More specifically, we provide the experimental result on its use to classical local frequency estimation methods GRR \cite{kairouz2016discrete} and OUE \cite{wang2017locally} as well as some other typical mean value estimation methods such as Duchi's method \cite{duchi2018minimax} and PM \cite{wang2019collecting}. We test all these methods on a synthetic data set with $10000$ participants following standard Gaussian  distribution. As shown in Fig. \ref{fig5}, the weighted aggregation has much improvement compared with directed aggregation for all the methods.

\begin{figure*}[ht]
     \centering
\subfigure[GRR]{
\label{Fig5-GRR}
\includegraphics[scale=0.25]{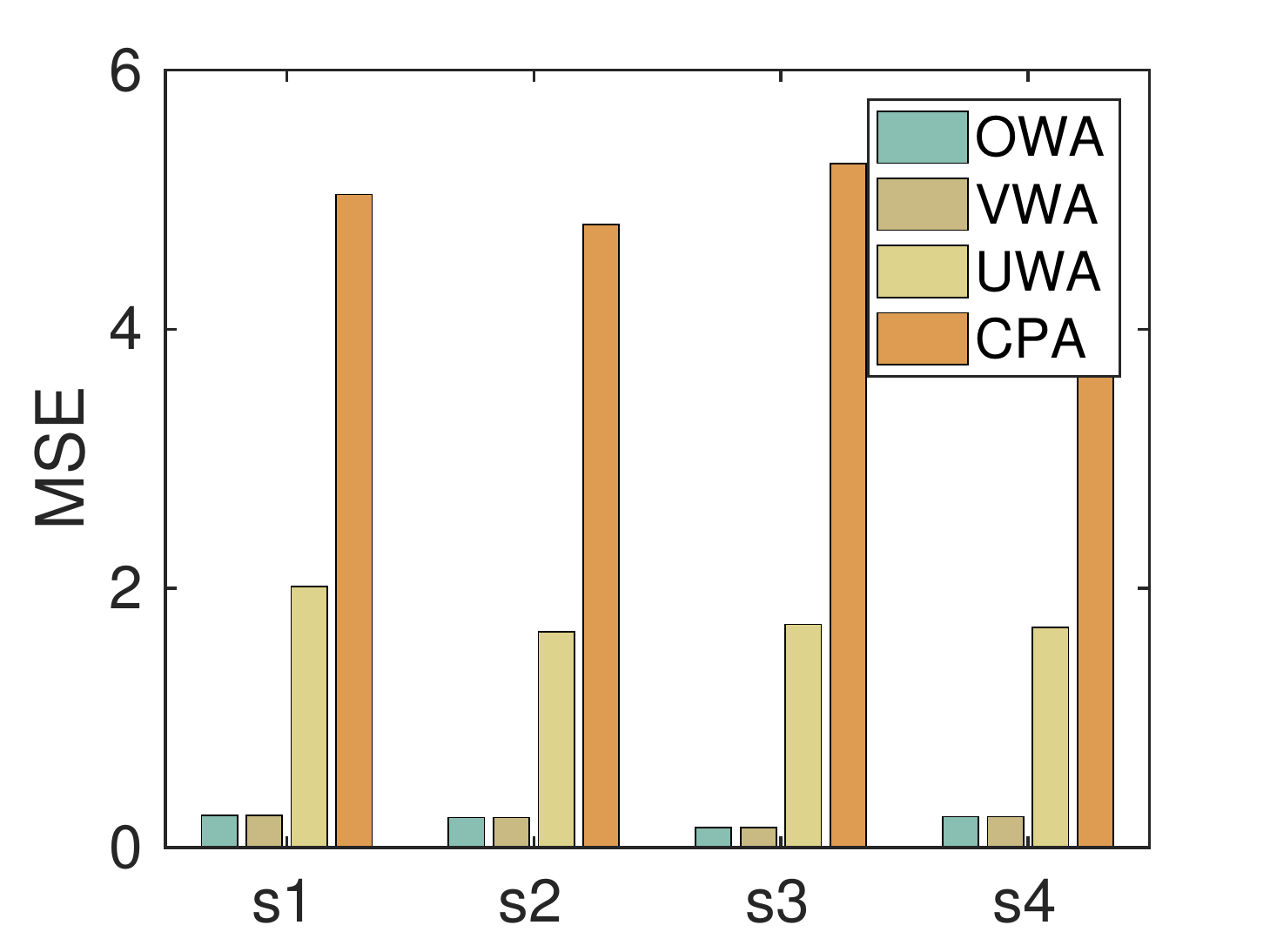}
}
\subfigure[OUE]{
\label{Fig5-OUE}
\includegraphics[scale=0.25]{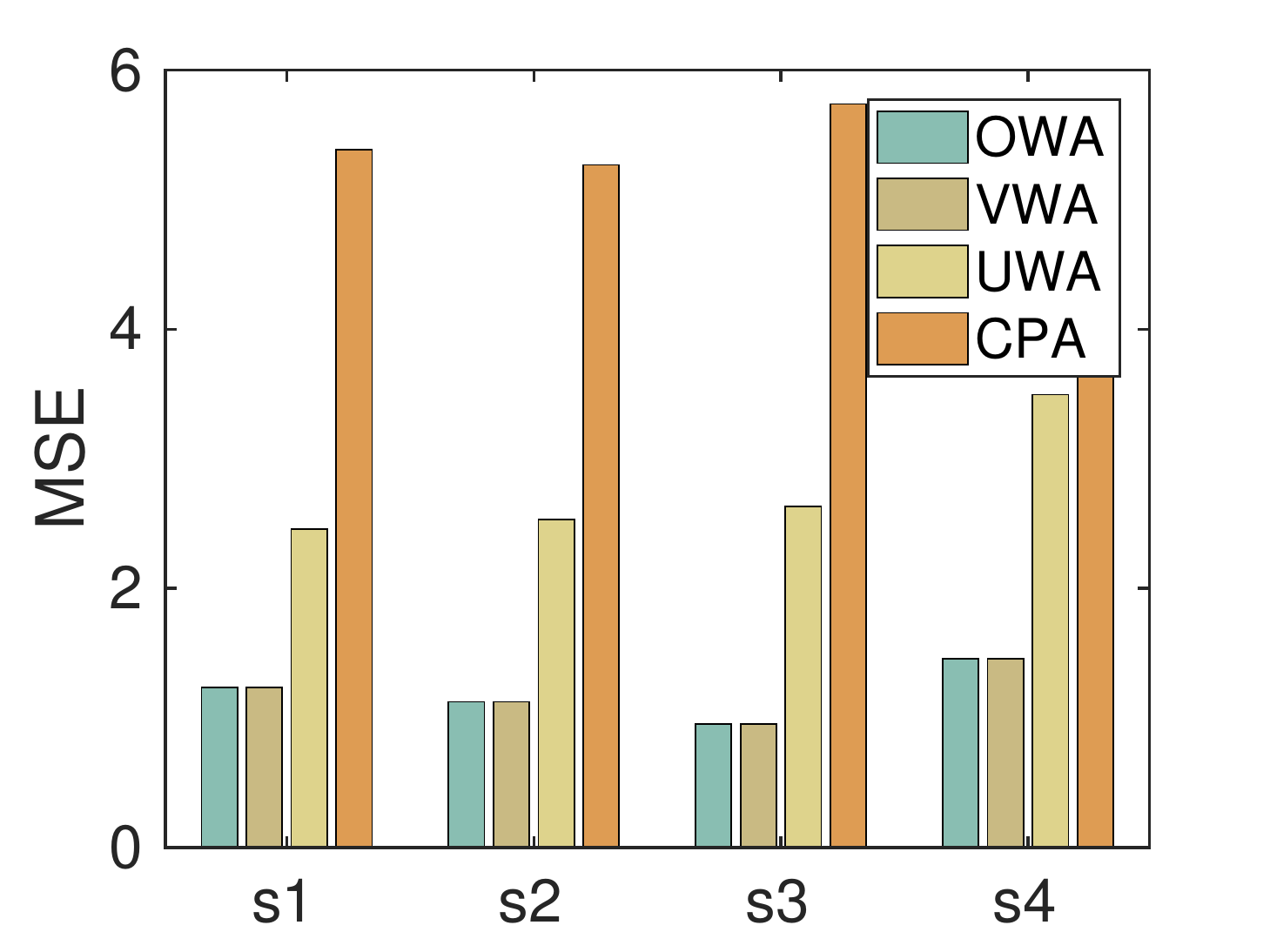}
}
\subfigure[Duchi]{
\label{Fig5-Duchi}
\includegraphics[scale=0.25]{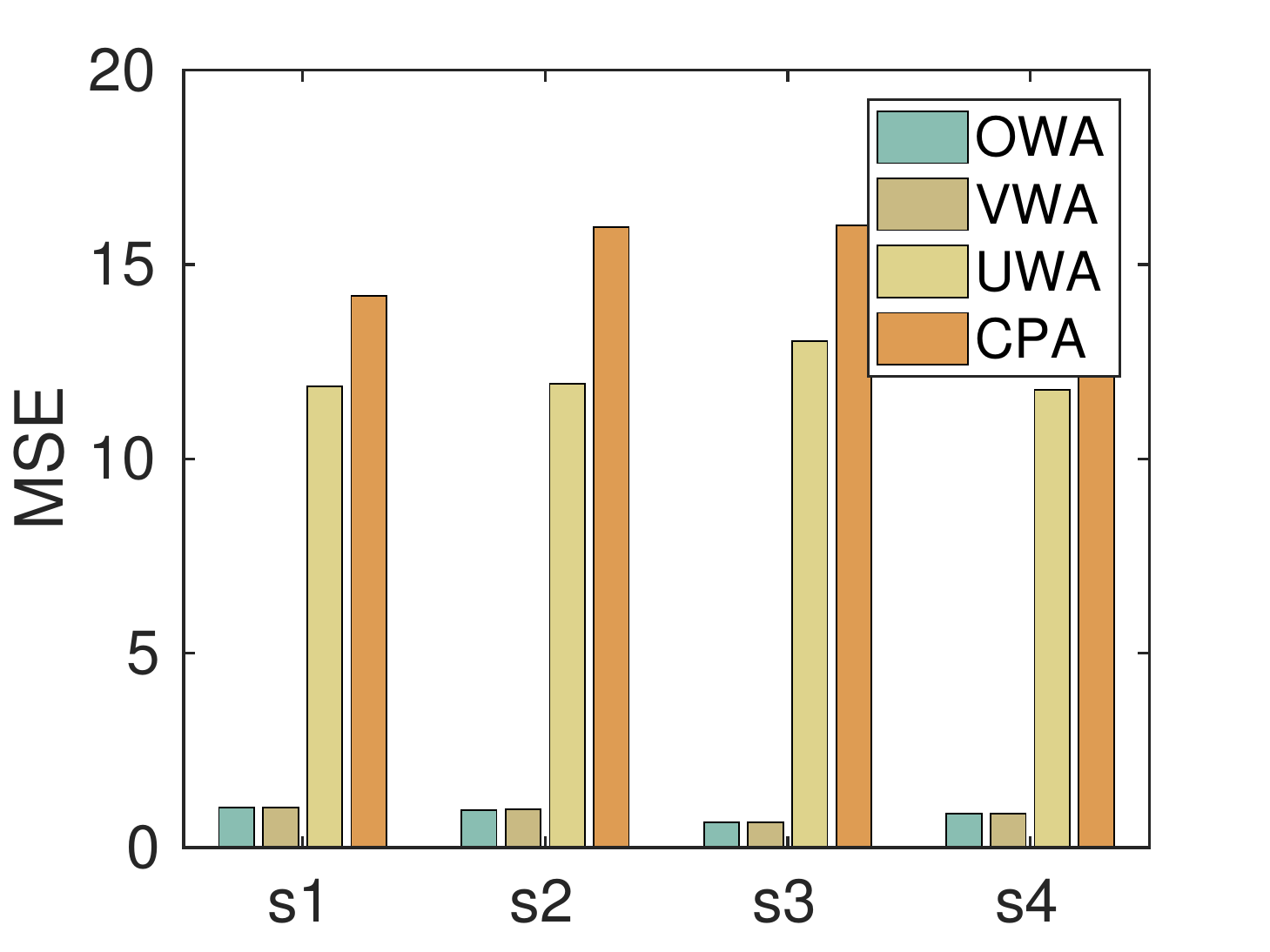}
}
\subfigure[PM]{
\label{Fig5-PM}
\includegraphics[scale=0.25]{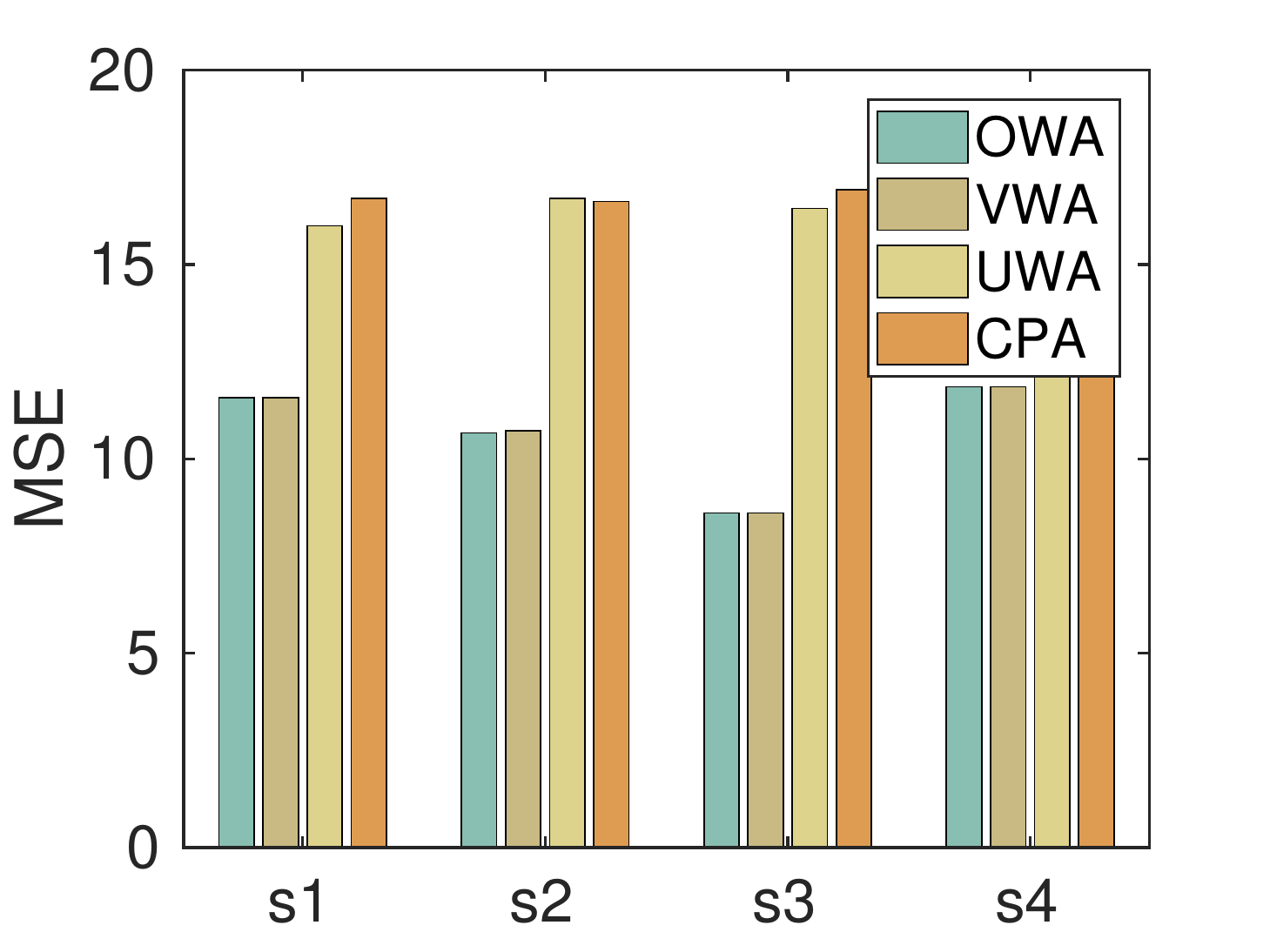}
}
\caption{Performance of weighted aggregation for other mechanisms}
\label{fig5}
\end{figure*}

\section{Related Work} \label{rela}

\noindent\textbf{Frequency estimation with local differential privacy.}
Frequency estimation mechanisms which provide local differential privacy against the server have been extensively studied in the literature~\cite{yang2020local}. One of the simplest methods to achieve this is the randomized response. 
Such technique was later generalized by Karionuz \textit{et al}.\cite{kairouz2016discrete} to be applicable for data with a higher dimensional attribute.
Later, Wang \textit{et al.} \cite{wang2017locally} proposed an optimized unary encoding (OUE) method, which encodes the true value following one hot encoding method, then performs randomized response to each bit of the vector. To deal with a higher dimensional attribute, Erlingsson \textit{et al.} \cite{erlingsson2014rappor} and Wang \textit{et al.} \cite{wang2017locally} proposed a hash-based method, which maps the user data into a much smaller domain before performing the randomized response process to reduce the statistical variance. Bassily and Smith \cite{bassily2015local} proposed a transformation-based method, which transforms the user's data from $d$ bits to only $1$ bit of data, which not only results in an estimation with smaller statistical variance, but also reduces the communication cost significantly. Besides, Wang \textit{et al.} \cite{wang2016mutual} proposed that instead of reporting the real items, users can randomly choose some of the items to report. Such subset selection method has a good performance in the intermediate privacy region compared to other methods. However, as we mentioned, all such local differential privacy methods need millions of participants to ensure the statistical accuracy due to the large variance introduced by the perturbation method.

\noindent\textbf{Combination of MPC and differential privacy.}
Multi-party computation provides a solution to securely compute functions on users' private data without disclosing them. However, MPC does not protect the privacy of any information that may be inferred from the output.
To protect the output privacy, it is natural to incorporate the differential privacy technique into the MPC calculation. 
In general, the combination of MPC and DP techniques that are considered in the literature utilize Laplace or Gaussian mechanism to provide the differential privacy guarantee. For example, Pettai and Laud \cite{pettai2015combining} studied and analyzed the overhead of adding Laplace noise to the secure multiparty computation. Bindschaedler et al. \cite{bindschaedler2017achieving} developed secure aggregation protocols to add Laplace and Gaussian noise to ensure differential privacy for the participants in a star network. 
The works in~\cite{guo2018practical,sun2020ldp} focus on the collaborative model training where Gaussian mechanism is incorporated to a training process to produce aggregated parameters that are differentially private. Furthermore, secure aggregation is introduced to reduce the required amount of the added noise. 
Besides, Hu et al. \cite{hu2021make} 
developed a secure computation protocol based on the Flajolet-Martin sketches for solving the Private Distributed Cardinality Estimation problem. In their work, they utilize the uncertainty introduced by the intrinsic estimation variance of the FM sketch to produce a differentially private output. 

\noindent\textbf{Sampling privacy.}
Sampling methods provide some measure of privacy-protection and are shown to perform better when the number of participants is small. It is usually used as an amplification method to strengthen the privacy bound in the distributed setting \cite{truex2020ldp, balle2020privacy}. Besides, Joy and Gerla \textit{et al.} \cite{joy2017differential} provided a sampling privacy mechanism under distributed environment. However, instead of discussing pure sampling mechanism, they let some of the data owners occasionally provide reports that are opposite to their supposed response. Husain \textit{et al.} \cite{husain2020local} assume that each user has multiple records and apply sampling process to each local data set. To our best knowledge, this paper is the first one to discuss the performance of pure sampling mechanism under distributed environment for frequency estimation. 

\noindent\textbf{Heterogeneous privacy setting.} Privacy-preserving schemes providing multiple levels of privacy setting are usually said to be in a personalized differential privacy setting, which are designed with the ability to provide personalized privacy to cater for users with different privacy requirements. 
The majority of the work \cite{yang2017personalized,wang2018personalized,zhang2019probabilistic,niu2020utility,cai2021profit} in this setting is focused on centralized setting and target on different applications, such as 
matrix factorization and crowdsourcing task assignment. On the other hand, distributed systems under personalized differential privacy setting have not been as extensively studied. Akter \textit{et al.} \cite{akter2017computing}. proposed a numerical perturbation method to estimate the data average under personalized differential privacy setting following the idea by Duchi \textit{et al.} \cite{duchi2018minimax}. However, in their schemes, the estimate from different privacy settings are directly aggregated together. Ye \textit{et al.} \cite{ye2019multiple} proposed a similar weighted aggregation for frequency estimation. However, the proposed weights only depend on the group size, which is not an optimal solution.

\section{Conclusion} \label{conc}

In this paper, we studied the frequency estimation problem and proposed a solution to provide a fair opportunity for companies with different scales to conduct users' data analysis. 
More specifically, our method is based on sampling. 
We first provided a theoretical bound for the sampling privacy and utility of the proposed method, then we extended the mechanism to the decentralized setting using secret sharing, which enables accurate analysis to be done without accessing user's original data. To make the mechanism more feasible for MPC, we proposed a two-stage sampling method, which further reduces the communication and computation load for users. In addition, we considered the scenario that users have different levels of privacy concern. Users are partitioned to several disjoint groups that operate independently to provide the frequency estimates for the respective groups. We proposed a weighted aggregation framework to obtain a more accurate estimate of the overall frequency from the estimates of different privacy groups.    

\bibliographystyle{IEEEtran}
\bibliography{IEEEabrv,references}

\end{document}